\documentclass[sigconf]{acmart}

\usepackage{IEEEtrantools}
\usepackage{enumerate}
\usepackage{xcolor}
\usepackage{amsthm}
\usepackage{graphicx}
\graphicspath{{./}}
\usepackage[caption=false,font=footnotesize]{subfig}
\usepackage{bm}
\usepackage{pifont}  
\usepackage{mathrsfs}

\usepackage{algorithm,algorithmic}

\usepackage{physics}

\newcommand{\nop}[1]{}

\newtheorem{example}{Example}
\newtheorem{theorem}{Theorem}
\newtheorem{corollary}{Corollary}
\newtheorem{proposition}{Proposition}
\newtheorem{remark}{Remark}
\newtheorem{assumption}{Assumption}

\usepackage{booktabs} 

\setcopyright{none}

\acmDOI{...}

\acmISBN{...}

\acmConference[]{}{}{}
\acmYear{2018}
\copyrightyear{2018}

\acmArticle{...}
\acmPrice{...}

\begin{document}

\title{A Hidden Resource in Wireless Channel Capacity:\\ Dependence Control in Action}

\author{Fengyou Sun and Yuming Jiang}
\affiliation{%
  \institution{Department of Information Security and Communication Technology\\ 
NTNU -- Norwegian University of Science and Technology}
  \city{Trondheim}
  \country{Norway}
}
\email{sunfengyou@gmail.com; ymjiang@ieee.org}

\begin{abstract}
This paper aims to initiate the research on dependence control, which transforms the dependence structure of a stochastic process in the system through dependence manipulation, to improve the system performance. Specifically, we develop a dependence control theory for wireless channels, focusing on three principles in dependence control: (i) the asymptotic decay rates of delay and backlog in the system are the measures for dependence comparison and ordering, (ii) the dependence in the arrival process and the service process have a dual potency to influence the system performance, and (iii) the manipulation of the dependence in the free dimensions of the arrival or service process transforms the dependence structure of the arrival or service process.
In addition, we apply the theory to the Markov additive process, which is a general model for a class of arrival processes and a versatile model for wireless channel capacity, and derive a set of results for various performance measures, including delay, backlog, and delay-constrained capacity. 
To demonstrate the use of the theory, we focus on dependence manipulation in wireless channel capacity, where we use copula to represent the dependence structure of the underlying Markov process of wireless channel capacity. We show that, based on a priori information of the temporal dependence of the uncontrollable parameters and the spatial dependence between the uncontrollable and controllable parameters, we can construct a sequence of temporal copulas of the Markov process and obtain a sequence of transition matrices of the controllable parameters to achieve the demanded dependence properties of the wireless channel capacity.
This dependence manipulation technique is validated by simulation.
\end{abstract}

%
%
%


\settopmatter{printacmref=false}
\renewcommand\footnotetextcopyrightpermission[1]{}
\pagestyle{plain}
\settopmatter{printfolios=true}
\maketitle

\section{Introduction}
\label{introduction}

Wireless communication has been around for over a hundred years, starting with Marconi's successful demonstration of wireless telegraphy in 1896 and transmission of the first wireless signals across the Atlantic in 1901 \cite{niehenke2014wireless}.
For cellular systems, the first generation is deployed in around 1980s \cite{niehenke2014wireless}, i.e., 1G, then 2G in 1990s, 3G in 2000s, 4G in 2010s, and the coming 5G in 2020s \cite{andrews2014will}.
It has become a trend that a new generation of wireless systems is deployed every new decade and the theme of each generation is to increase the capacity and spectral efficiency of wireless channels. 
The trend is driven by the explosion of wireless traffic that is a rough reflection of people's demand on wireless communication, and the paradox of supply and demand \cite{hecht2016bandwidth} is kept relieving generation by generation through exploiting the physical resources, i.e., power, diversity, and degree of freedom \cite{tse2005fundamentals}.
Considering trillions of devices to be connected to the wireless network, high capacity demand, and stringent latency requirement in the coming 5G \cite{andrews2014will}, it is imperative to rethink the wireless channel resources.
In a recent paper \cite{sun2017statistical}, the stochastic dependence in wireless channel capacity is identified as a hidden resource in achieving delay guarantee. Specifically, if the wireless channel capacity bears negative dependence, the wireless channel can even attain a better delay performance under a smaller capacity  \cite{sun2017statistical}. 

This paper aims to initiate the research on dependence control, which is still blank in the literature of both stochastic models and wireless communication. It complements the dependence modeling field \cite{denuit2006actuarial,ruschendorf2013mathematical,mcneil2015quantitative}, and provides a potential for the development of new wireless technologies.
Specifically in this paper, we develop a dependence control theory for wireless channels, 
and the results are focused on the following three principles in dependence control.
\begin{enumerate}
\item\label{contribution-first}
The wireless channel performance is reflected by the asymptotic decay rate of the tail of delay or backlog, which is used in identifying, comparing, and ordering the stochastic dependence. 

We provide exact expressions of the asymptotic decay rates of delay and backlog, based on some general assumptions of the arrival process and service process, which are capable of characterizing weak forms of dependence and light-tailed process.

\item\label{contribution-second}
The dependence in the arrival process and the service process have a dual potency for manipulating the queueing process of the wireless channel and the asymptotic decay identity. 

We prove that the dependence ordering of the arrival process or the service process results in correspondingly the ordering of the queue increment process or the negative queue increment process, based on the assumptions, which further results in the ordering of the asymptotic decay identities. This manifestation is optimized for constant arrival process or service process with the same mean as their random peers.

\item\label{contribution-third}
The manipulation of the free dimensions of the arrival process or the service process is able to transform the dependence structure of the arrival process or the service process.

We provide a functional perspective on the wireless channel capacity, i.e., the capacity process is treated as a functional of a multivariate stochastic process, composing of controllable random parameters and uncontrollable random parameters. Particularly, we prove that the manipulation of the dependence in the controllable random parameters results in the dependence transformation in the capacity process. 
For the arrival process, we prove that the dependence manipulation of the individual process results in the dependence transformation of the aggregated process in both deterministic multiplexing and random multiplexing.
Particularly, the dependence manipulation and transformation are independent of the assumptions on the dependence forms and tail behaviors.
\end{enumerate}

As an application of the dependence control theory and as a justification to the assumptions in this paper, we apply the theory to the Markov additive process, which is capable of characterizing a large class of arrival processes and is versatile in capturing the dependence in the wireless channel capacity. In this paper, the focus of the dependence manipulation is on the wireless channel capacity.
The results in this application are summarized as follows.

\begin{enumerate}
\setcounter{enumi}{3}
\item\label{application-first}
For the performance measures of the wireless channel, we provide the non-asymptotic and time-dependent performance results of delay and backlog, and an upper bound of the delay-constrained capacity, for Markov additive arrival process and capacity process.

The delay and backlog results in the Markov additive model extend the analysis for constant arrivals in \cite{sun2017statistical} to general arrivals, and extend the scenario of stationary processes in \cite{poloczek2015service} to non-stationary processes, and complementary double-sided bounds are provided.
In addition to the probability on infinite time horizon, time-dependent results are also given on finite time horizon.
The delay-constrained capacity is an extension of the result for constant arrivals in \cite{sun2017statistical}.

\item\label{application-second}
For the dependence manipulation of the capacity, we treat the underlying Markov process as a multi-dimensional process of the controllable and uncontrollable random parameters, and use copula to represent the Markov property and to configure the transition matrix.
The copula manipulation technique is validated by simulation.

We model the random parameters in wireless channel capacity as a multivariate Markov process. The Markov family copula in \cite{darsow1992copulas,overbeck2015multivariate} are used not only as a mechanism for dependence modelling, 
but also as a tool for dependence controlling.
We apply the no-Granger causality to model the relationship between the controllable and uncontrollable parameters, and the sufficient and necessary condition for Markov process is extended from the bivariate case in \cite{cherubini2011copula} to the multivariate case in this paper.
Note that, the copula property of Markov process investigated in \cite{darsow1992copulas} is extended to high order case in \cite{ibragimov2009copula} and multivariate case in \cite{overbeck2015multivariate}. No-Granger causality is a concept in econometrics and its relation with Markov process is investigated in \cite{cherubini2011copula}.
\end{enumerate}

In all, the dependence control theory composing of three principles (\ref{contribution-first}) to (\ref{contribution-third}), and the application results (\ref{application-first}) to (\ref{application-second}), constitutes the main contributions of this paper. 
To build the dependence control theory, we adopt a few mathematical tools. Specifically, change of measure is used to find the asymptotic decay identities and to explain the dual potency of arrival and service dependence, and stochastic order is used to prove the efficiency of dependence manipulation and to compare different processes. 
In application, martingale is used in the performance analysis, and copula is used to represent the Markov property and the no-Granger causality, and as a dependence manipulation technique.

The remainder of this paper is structured as follows. 
In Sec. \ref{preliminaries}, we introduce some basic concepts of wireless channel capacity, the queueing behavior of the wireless channel, and the assumptions that are used in this paper.
In Sec. \ref{dependence-control}, we present the dependence control theory, including the measure for dependence comparison and ordering, the dual potency of the arrival and service for dependence control, and the transform of dependence structure, which are termed as the three principles of dependence control in this paper. Particularly, the first principle is based on the assumption, the second principle partially relies on the assumption, and the third principle is independent of the assumption.
In Sec. \ref{application}, we provide an application of the dependence control theory, where the Markov additive process is a concrete justification of the assumptions in this paper. 
Finally, we conclude the paper in Sec. \ref{conclusion}.

\section{Preliminaries}
\label{preliminaries}

In this section, we present some basic concepts of wireless channel capacity, the queueing behavior of the wireless system, and the assumption in this paper.

\subsection{Basic Concepts}

Consider a flat fading channel with input $x(t)$, output $y(t)$, fading process $h(t)$, and additive white Gaussian noise (AWGN) $n(t)\sim\mathcal{C}\mathcal{N}(0,N_0)$,
the complex baseband representation is expressed as \cite{goldsmith2005wireless,tse2005fundamentals}
\begin{equation}
y(t) = h(t)x(t) + n(t),
\end{equation}
conditional on a realization of $h(t)$, the mutual information 
is expressed as \cite{goldsmith2005wireless} 
\begin{IEEEeqnarray}{rCl}\label{eq-2}
I(X;Y|h(t)) = \sum\limits_{x\in\mathcal{X},y\in\mathcal{Y}} \mathbb{P}(x,y|h_t)\log_2\frac{\mathbb{P}(x,y|h_t)}{\mathbb{P}(x|h_t) \mathbb{P}(y|h_t)}, \IEEEeqnarraynumspace
\end{IEEEeqnarray}
where $\mathcal{X}$ and $\mathcal{Y}$ are respectively the input and output alphabets of the channel. 
For multiple input and multiple output channel, the generalized formula is available in \cite{telatar1999capacity,foschini1998limits}.

The maximum mutual information over input distribution at $t$, denoted as $C(t)$, is defined as { instantaneous capacity} \cite{costa2010multiple}:
\begin{equation}\label{eq-1}
C(t) = \max\limits_{ \mathbb{P}(x)}I(X;Y|h(t)),
\end{equation}
where $\mathbb{P}(x) \equiv \mathbb{P}(X=x)$, $x\in\mathcal{X}$.
The sum of instantaneous capacity in discrete time $(s,t]$,
denoted as $S(s,t)$, is defined as { cumulative capacity}:
\begin{eqnarray}
S(s,t) = \sum\limits_{i=s+1}^{t}{C(i)}
\end{eqnarray}
Denote $S(t)\equiv S(0,t)$.
The time average of the cumulative capacity through $(0,t]$ is defined as { transient capacity}:
\begin{equation}
\overline{C}(t) = \frac{S(t)}{t}. 
\end{equation}

\begin{example}
For a single input single output channel, 
if the channel side information is only known at the receiver, the instantaneous capacity is expressed as \cite{tse2005fundamentals}
\begin{equation}\label{def-ic}
C(t) = W\log_{2}\left(1+\Gamma|h(t)|^{2}\right),
\end{equation}
where $|h(t)|$ denotes the envelope of $h(t)$, $\Gamma = {P}/{N_{0}W}$ denotes the average received SNR per complex degree of freedom, $P$ denotes the average transmission power per complex symbol, $N_{0}/2$ denotes the power spectral density of AWGN, and $W$ denotes the channel bandwidth. 
\end{example}

The following theorem presents a fundamental property of wireless channel capacity, which is available in \cite{sun2017statistical} and restated here.

\begin{theorem}
For flat fading, the instantaneous capacity is expressed as the logarithm transform of the instantaneous channel gain, i.e., $C(t) = W\log_{2}(1+\Gamma{h(t)^2})$, $\forall t$. If the distribution of the fading process is not heavier than fat tail, the distribution of the instantaneous capacity is light-tailed. 
Specifically, if a wireless channel is Rayleigh, Rice, Nakagami-$m$, Weibull, or lognormal fading channel, its instantaneous capacity distribution is light-tailed. 
\end{theorem}

\subsection{Queueing Behavior}

The wireless channel is essentially a queueing system with cumulative service process $S(t)$ and cumulative arrival process $ A(0,t)=\sum\limits_{s=1}^{t}a(s)$, 
where $a(t)$ denotes the traffic input to the channel at time slot $t$,
and the temporal increment in the system is expressed as
\begin{equation}
X(t) = a(t)-C(t).
\end{equation} 
The queueing behavior of the wireless channel is expressed through the backlog in the system, which is a reflected process of the temporal increment $X(t)$ \cite{asmussen2003applied}, i.e.,
\begin{equation}
B(t+1) = \left[ B(t) + X(t)\right]^{+}.
\end{equation}
Assume $B(0)=0$,
the backlog function is then expressed as
\begin{equation}\label{eq-bl}
B(t) = \sup_{0\le{s}\le{t}}({A}(s,t)-{S}(s,t)). 
\end{equation}
For a lossless system, the output is the difference between the input and backlog, 
\begin{equation}
A^{\ast}(t) = A(t) - B(t), 
\end{equation}
and the delay is defined via the input-output relationship, i.e., 
\begin{equation}
D(t) = \inf\left\{ d\ge{0}: A(t-d)\le A^\ast(t) \right\},
\end{equation}
which is the virtual delay that a hypothetical arrival has experienced on departure.
The maximum rate of traffic with delay requirement that the system can support without dropping is defined as the delay-constrained capacity or throughput \cite{sun2017statistical}: 
\begin{equation}
\overline{C}{(d,\epsilon)} = \sup_{P(D(t) > d) \le \epsilon, \forall t} E\left[ \frac{A(t)}{t} \right].
\end{equation}

The delay tail probability is expressed as
\begin{IEEEeqnarray}{rCl}
\mathbb{P} ( D > d ) &=& \mathbb{P} \left\{ A(t-d) > A^\ast(t) \right\} \\
&=& \mathbb{P} \left\{ A(t-d) > \inf_{0\le{s}\le{t}} \left\{ A(0,s) + S(s,t) \right\}  \right\} \\
&=& \mathbb{P} \left\{ \sup_{t\ge{d}}\left\{ A(d,t) - S(0,t) \right\} > 0 \right\},
\end{IEEEeqnarray}
where the last step follows time reversibility. 
The backlog tail probability is expressed as
\begin{IEEEeqnarray}{rCl}
\mathbb{P} ( B > b ) &=& \mathbb{P}\left\{ \sup_{0\le{s}\le{t}} (A(s,t)-S(s,t)) > b \right\} \\
&=& \mathbb{P} \left\{ \sup_{t\ge{0}} (A(t)-S(t)) > b \right\},
\end{IEEEeqnarray}
where the second equality follows time reversibility.

\subsection{Assumption}

We specify the cumulent generating function of the cumulative arrival process $A(t)$, the cumulative service process $S(t)$, and the increment process of the queue $A(t)-S(t)$.

The assumption for the queue increment process is as follows \cite{glynn1994logarithmic}, without assumption on the dependence between the arrival process and service process. 

\begin{assumption}\label{assumption}
Denote $\mathfrak{S}(t)=A(t)-S(t)$ and $X(t)=a(t)-C(t)$. Assume that there exist $\gamma,\epsilon>0$ such that
\begin{enumerate}
\item\label{aspt-1}
$\kappa_{t}(\theta) = \log \mathbb{E} e^{\theta\mathfrak{S}(t)}$ is well-defined and finite for $\gamma-\epsilon<\theta<\gamma+\epsilon$;
\item\label{aspt-2}
$\limsup_{t\rightarrow\infty}\mathbb{E}e^{\theta X(t)}<\infty$ for $-\epsilon<\theta<\epsilon$;
\item\label{aspt-3}
$\kappa(\theta)=\lim_{t\rightarrow\infty}\frac{1}{t}\kappa_t(\theta)$ exists and is finite for $\gamma-\epsilon<\theta<\gamma+\epsilon$;
\item\label{aspt-4}
$\kappa(\gamma)=0$ and $\kappa$ is differentiable at $\gamma$ with $0<\dot\kappa(\gamma)<\infty$.
\end{enumerate} 
\end{assumption}

A justification to the above assumption is the following proposition \cite{glynn1994logarithmic}, with independence assumption between the arrival and service process.

\begin{proposition}\label{proposition-arrival-service}
Assume independence between the sequences of $a(t)$ and $C(t)$, $t\ge 0$. Let $\gamma,\epsilon>0$ be as in Assumption such that
\begin{enumerate}
\item\label{proposition-1}
$\kappa^{A}_{t}(\theta) = \log \mathbb{E} e^{\theta{A}(t)}$ is well-defined and finite for $\gamma-\epsilon<\theta<\gamma+\epsilon$;
\item
$\limsup_{t\rightarrow\infty}\mathbb{E}e^{\theta a(t)}<\infty$ for $-\epsilon<\theta<\epsilon$;
\item
$\kappa^A(\theta)=\lim_{t\rightarrow\infty}\frac{1}{t}\kappa^A_t(\theta)$ exists, is differentiable at $\gamma$, and is finite for $\gamma-\epsilon<\theta<\gamma+\epsilon$;
\item
$\kappa^{-S}_{t}(\theta) = \log \mathbb{E} e^{-\theta{S}(t)}$ is well-defined and finite for $\gamma-\epsilon<\theta<\gamma+\epsilon$;
\item
$\limsup_{t\rightarrow\infty}\mathbb{E}e^{-\theta C(t)}<\infty$ for $-\epsilon<\theta<\epsilon$;
\item
$\kappa^{-S}(\theta)=\lim_{t\rightarrow\infty}\frac{1}{t}\kappa^{-S}_t(\theta)$ exists, is differentiable at $\gamma$, and is finite for $\gamma-\epsilon<\theta<\gamma+\epsilon$;
\item
$\kappa(\theta) = \kappa^A(\theta) + \kappa^{-S}(\theta)$.
\end{enumerate}
\end{proposition}

The Assumption and Proposition apply to the scenario \cite{asmussen2010ruin}, where there are weak forms of dependence, e.g., Markov dependence, and the average of the cumulent generating function exists and converges, e.g., light-tailed process.

\section{Dependence Control}
\label{dependence-control}

This section focuses on the dependence control theory, including the dependence identification measure, the dual potency, and the dependence transform, which are the three principles of the dependence control in this paper.

\subsection{Measure Identity}

In this subsection, we derive the asymptotic decay rate of delay and backlog, which are the fundamental measures for dependence comparison and control in this paper.

We introduce a change of measure for $a_1, \ldots, a_n, c_1, \ldots, c_n$, i.e.,
\begin{multline}
\widetilde{F}_{n}\left( d a_1, \ldots, d a_n, d c_1, \ldots, d c_n \right) \\
= e^{\gamma s_n - \kappa_n(\gamma)}{F}_{n}\left( d a_1, \ldots, d a_n, d c_1, \ldots, d c_n \right),
\end{multline}
where $F_n$ is the distribution of $a_1, \ldots, a_n, c_1, \ldots, c_n$ and $s_n = a_1 - c_1 + \ldots + a_n - c_n$. 

Assume independence between $a_1, \ldots, a_n$ and $ c_1, \ldots, c_n$, 
then the distributions of $a_1, \ldots, a_n$ and $c_1, \ldots, c_n$ in the new probability measure are given by
\begin{IEEEeqnarray}{rCl}
\widetilde{F}^{A}_{n}\left( d a_1, \ldots, d a_n \right) &=& e^{\gamma s^A_n - \kappa^A_n(\gamma)}{F}_{n}\left( d a_1, \ldots, d a_n, \bm{1} \right), \\
\widetilde{F}^{S}_{n}\left( d c_1, \ldots, d c_n \right) &=& e^{\gamma s^{-S}_n - \kappa^{-S}_n(\gamma)}{F}_{n}\left(\bm{1}, d c_1, \ldots, d c_n \right),
\end{IEEEeqnarray}
where $s^A_n = a_1 + \ldots + a_n$, $s^{-S}_n = -( c_1 + \ldots + c_n )$, and 
\begin{equation}
\kappa_n(\gamma) = \kappa^A_n(\gamma) + \kappa^{-S}_n(\gamma).
\end{equation}

We present the asymptotic decay rate of delay and backlog in the following theorem, which shows that the asymptotic behavior of the tail is exponential for weak forms of dependence and light-tailed process that are indicated by Assumption \ref{assumption} and Proposition \ref{proposition-arrival-service}. We prove the theorem in Appendix \ref{proof-of-theorem-decay-rate-delay-backlog}.

\begin{theorem}\label{theorem-decay-rate-delay-backlog}
Under the conditions in Proposition \ref{proposition-arrival-service},
the asymptotic decay rate of delay is
\begin{equation}
\lim_{d\rightarrow\infty}\frac{1}{d}{\log \mathbb{P}(D> d)} = -\kappa^{A}(\theta),
\end{equation}
and the asymptotic decay rate of backlog is 
\begin{equation}
\lim_{b\rightarrow\infty}\frac{1}{b}{\log \mathbb{P}( B> b )} = -\theta,
\end{equation}
where $\theta>0$ is the root to the stability equation
\begin{equation}
\kappa^A(\theta) + \kappa^{-S}(\theta) = 0.
\end{equation}
\end{theorem}

The following theorem presents a property of the time average of a process in the original probability measure.

\begin{theorem}
Under the conditions in Assumption \ref{assumption} or in Proposition \ref{proposition-arrival-service}.
For the process $M(t)$, i.e., $A(t)$, $S(t)$, or $A(t)-S(t)$, we have
\begin{equation}
\lim_{t\rightarrow \infty} \mathbb{E} \left[ \frac{M(t)}{t} \right] = \dot\kappa(0).
\end{equation}
\end{theorem}

\begin{proof}
Considering
\begin{equation}
\kappa(\theta) = \lim_{t\rightarrow\infty}\frac{1}{t}\log\mathbb{E} e^{\theta M(t)},
\end{equation}
calculate the derivative,
\begin{equation}
\dot\kappa(\theta) = \lim_{t\rightarrow \infty}\frac{1}{t}\left[ \frac{1}{\mathbb{E} e^{\theta M(t)}} \mathbb{E}\left[ e^{\theta M(t)} M(t) \right] \right],
\end{equation}
let $\theta=0$, 
the result then follows.
\end{proof}

The proof of Theorem \ref{theorem-decay-rate-delay-backlog} requires some preliminary results in the new probability measure, which are of independent interest for understanding the properties of the new probability measure.

The following result shows that, in the new probability measure, for fixed $d, k\in \mathbb{N}$, $\frac{ \mathfrak{S}(d,n-k)}{n}$ converges in probability to $\dot\kappa(\gamma)$, 
which indicates that the convergence is insensible to the head and tail of the sequence $ \mathfrak{S}(n)$. 
Note, in the original probability measure, $\frac{ \mathfrak{S}(n)}{n}$ converges in probability to $\dot\kappa(0)$, and there is a sign change from the original to the new probability measure.

\begin{theorem}
Let $n, d, k\in\mathbb{N}$ and $d, k<\infty$, and $\widetilde{\mu} \equiv \dot\kappa(\gamma)$. Then,
\begin{equation}
\lim_{n\rightarrow\infty}\widetilde{\mathbb{P}}_{n} \left( \left| \frac{ \mathfrak{S}(d,n-k)}{n} -\widetilde{\mu} \right| > \eta \right) = 0,\ \forall \eta>0.
\end{equation}
\end{theorem}

An equivalent expression of the above theorem is the following theorem,
which is proved in Appendix \ref{proof-of-theorem-equivalent-converge-in-probability}.
Particularly, in the proof, no assumption of time reversibility or stationary is needed and the bounding function $z^n$ is set to facilitate the proof of Theorem \ref{theorem-decay-rate-delay-backlog}. 

\begin{theorem}\label{theorem-equivalent-converge-in-probability}
Let $n, d, k\in\mathbb{N}$ and $d, k<\infty$, and $\widetilde{\mu} \equiv \dot\kappa(\gamma)$. For each $\eta>0$, there exist $z\equiv z(\eta)\in(0,1)$ and $n_0$ such that
\begin{equation}
\widetilde{\mathbb{P}}_n \left( \left| \frac{ \mathfrak{S}(d,n-k) }{n} -\widetilde{\mu} \right| > \eta \right) \le z^n,\ \text{for } n\ge n_0.
\end{equation}
\end{theorem}

Replace $\mathfrak{S}(d)$ with $A(d)$ and a finite constant $x$, we get the following corollary, which is used in the proof of Theorem \ref{theorem-decay-rate-delay-backlog}.

\begin{corollary}\label{corollary-for-decay}
Let $n, d, k\in\mathbb{N}$ and $d, k<\infty$,  and $\widetilde{\mu} \equiv \dot\kappa(\gamma)$. Let $x\in \mathbb{R}$ and $x<\infty$. For each $\eta>0$, there exist $z\equiv z(\eta)\in(0,1)$ and $n_0$ such that
\begin{equation}
\widetilde{\mathbb{P}}_n \left( \left| \frac{ \mathfrak{S}(n-k)-A(d) + x }{n} -\widetilde{\mu} \right| > \eta \right) \le z^n,\ \text{for } n\ge n_0.
\end{equation}
\end{corollary}

\begin{proof}
The corollary is a deduction of the above theorem. 
\end{proof}

The following theorem shows the ultimate throughput of the wireless channel for a channel specification.

\begin{theorem}\label{theorem-delay-constrained-capacity}
Under the conditions in Proposition \ref{proposition-arrival-service},
the delay-constrained capacity is asymptotically upper bounded by
\begin{equation}
\lim_{ \substack{ d\rightarrow\infty\\ \epsilon\rightarrow 0} } \overline{C}(d,\epsilon) \le \frac{\kappa^{A}(\theta)}{\theta},
\end{equation}
where $\theta = \{ \theta>0: \kappa^{A}(\theta) + \kappa^{-S}(\theta) = 0 \}$.
\end{theorem}

\begin{proof}
Considering $\kappa^A(\theta) = \lim_{t\rightarrow\infty}\frac{1}{t}\log \mathbb{E}e^{\theta A(t)}$, take a logarithm and according to Jensen's inequality, we get
\begin{equation}
\lim_{t\rightarrow\infty} \mathbb{E}\left[ \frac{A(t)}{t} \right] \le \frac{\kappa^A(\theta)}{\theta},\ \forall\theta>0,
\end{equation}
since $d\rightarrow\infty$ implies $t\rightarrow\infty$,
\begin{equation}
\lim_{ \substack{ d\rightarrow\infty\\ \epsilon\rightarrow 0} } \overline{C}(d,\epsilon) \le \lim_{ \substack{ d\rightarrow\infty\\ \epsilon\rightarrow 0} } \sup_{ \mathbb{P}(D>d)\le\epsilon } \frac{\kappa^A(\theta)}{\theta},
\end{equation}
which is optimized at $\theta = \{ \theta>0: \kappa^{A}(\theta) + \kappa^{-S}(\theta) = 0 \}$, 
since $\dot{\kappa}^{A}(\theta)\ge \frac{\kappa^{A}(\theta)}{\theta} \ge 0$, $\forall \theta >0$,
where the first step follows that $\kappa^A(x) \ge \kappa^A(y) + \dot\kappa^A(y)(x-y)$ for all $x$ and $y$ satisfying the Proposition \ref{proposition-arrival-service} \cite{boyd2004convex}.
\end{proof}

\begin{remark}
Since the inequality, $\lim\limits_{t\rightarrow\infty} \mathbb{E}\left[ \frac{A(t)}{t} \right] \le \frac{\kappa^A(\theta)}{\theta},\ \forall\theta>0,$ has no specification on the constraints, besides the delay-constrained throughput $\overline{C}(d,\epsilon) \le \sup_{ \mathbb{P}(D>d)\le\epsilon } \frac{\kappa^A(\theta)}{\theta}$, it is interesting to investigate the backlog-constrained throughput $\overline{C}(b,\xi) \le \sup_{ \mathbb{P}(B>b)\le\xi } \frac{\kappa^A(\theta)}{\theta}$, and the backlog-and-delay-constrained throughput $\overline{C}(b,d,\epsilon, \xi) =\\ \min (\overline{C}(d,\epsilon), \overline{C}(b,\xi))$.
\end{remark}

\subsection{Dual Potency}

In this subsection, 
we show that the dependence in the arrival process and in the service process, both have an impact on the queue performance, in case the dependence manipulation in the arrival process is not available, we can transfer to the dependence manipulation in the service process, vice versa.
We use the standard definition of convex order $\le_{cx}$ and supermodular order $\le_{sm}$ in \cite{muller2002comparison}\cite{shaked2007stochastic}, i.e., random variables $X \le_{cx} Y$ if $\mathbb{E}[\phi(X)] \le \mathbb{E}[\phi(Y)]$ for all convex function $\phi$, and random vectors $\bm{X} \le_{sm} \bm{Y}$ if $\mathbb{E}[\phi(\bm{X})] \le \mathbb{E}[\phi(\bm{Y})]$ for all supermodular function $\phi$.

In the previous subsection, we consider the change of measure, with the increment process $\mathfrak{S}(t) = A(t)-S(t)$,
\begin{multline}
\widetilde{F}_{n}\left( d a_1, \ldots, d a_n, d c_1, \ldots, d c_n \right) \\
= e^{\gamma \mathfrak{S}_n - \kappa^{\mathfrak{S}}_n(\gamma)}{F}_{n}\left( d a_1, \ldots, d a_n, d c_1, \ldots, d c_n \right),\ \gamma>0,
\end{multline}
on the positive part of the parameter axis,
the stability equation is expressed as
\begin{equation}
\kappa^{A}(\theta) + \kappa^{-S}(\theta) = 0,\ \theta>0, 
\end{equation}
the optimal $\gamma = \{\theta>0: \kappa^{A}(\theta) + \kappa^{-S}(\theta) =0 \}$ is the asymptotic decay rate of backlog and the $-\kappa^{-S}(\gamma) = \kappa^{A}(\gamma)$ is the asymptotic decay rate of delay.
This change of measure, with the negative increment process $-\mathfrak{S}(t) = S(t)-A(t)$, has an equivalent expression
\begin{multline}
\widetilde{F}_{n}\left( d a_1, \ldots, d a_n, d c_1, \ldots, d c_n \right) \\
= e^{-\gamma \mathfrak{S}_n - \kappa^{-\mathfrak{S}}_n(\gamma)}{F}_{n}\left( d a_1, \ldots, d a_n, d c_1, \ldots, d c_n \right),\ \gamma<0,
\end{multline}
on the negative part of the parameter axis,
we get a dual expression of the stability equation
\begin{equation}
\kappa^{S}(\theta) + \kappa^{-A}(\theta) = 0,\ \theta<0, 
\end{equation}
the optimal $-\gamma = -\{\theta<0: \kappa^{S}(\theta) + \kappa^{-A}(\theta) =0 \}$ is the asymptotic decay rate of backlog and the $\kappa^{-A}(\gamma) = -\kappa^{S}(\gamma)$ is the asymptotic decay rate of delay.
It's a matter of taste to choose one or the other and each has a direct implication on the problem under consideration.

The sufficient conditions for the ordering of the asymptotic decay rates of delay and backlog are shown in the following theorem.

\begin{theorem}\label{theorem-ordering-arrival-and-service}
Let $\mathfrak{S}(t)=A(t)-S(t)$ and $\widetilde{\mathfrak{S}}(t)= \widetilde{A}(t) - \widetilde{S}(t)$.
Then
\begin{IEEEeqnarray}{rCl}
\mathfrak{S}(t) \le_{cx} \widetilde{\mathfrak{S}}(t),\ \forall{t}\in\mathbb{N}  
&\implies& 0 < \widetilde{\gamma} \le \gamma,\\
-\mathfrak{S}(t) \le_{cx} -\widetilde{\mathfrak{S}}(t),\ \forall{t}\in\mathbb{N} 
&\implies& 0 > \widetilde{\gamma} \ge \gamma,
\end{IEEEeqnarray}
particularly, if $\mathfrak{S}(t)$ and $\widetilde{\mathfrak{S}}(t)$ have identical service process $S(t)$, then
\begin{IEEEeqnarray}{rCl}
\mathfrak{S}(t) \le_{cx} \widetilde{\mathfrak{S}}(t),\ \forall{t}\in\mathbb{N}  
&\implies& 0 < \widetilde{\gamma} \le \gamma  \\
&\implies& \kappa^{-S}(\gamma) \le \widetilde{\kappa}^{-S}(\widetilde{\gamma}) < 0,
\end{IEEEeqnarray}
and if $\mathfrak{S}(t)$ and $\widetilde{\mathfrak{S}}(t)$ have identical arrival process $A(t)$, then
\begin{IEEEeqnarray}{rCl}
-\mathfrak{S}(t) \le_{cx} -\widetilde{\mathfrak{S}}(t),\ \forall{t}\in\mathbb{N} 
&\implies& 0 > \widetilde{\gamma} \ge \gamma \\
&\implies&  \kappa^{-A}(\gamma) \ge \widetilde{\kappa}^{-A}(\widetilde{\gamma}) > 0.
\end{IEEEeqnarray}
\end{theorem}

\begin{proof}
Since the exponential function is convex, $\mathfrak{S}(t) \le_{cx} \widetilde{\mathfrak{S}}(t)$ implies $\mathbb{E}e^{\theta \mathfrak{S}(t)} \le \mathbb{E}e^{\theta \widetilde{\mathfrak{S}}(t)}$, thus $\kappa_t(\theta) \le \widetilde{\kappa}_t({\theta})$ and subsequently $\kappa(\theta) \le \widetilde{\kappa}({\theta})$, $\forall{\theta}\in\mathbb{R}$.
Since $\kappa$ is convex with $\kappa(0) = \kappa(\gamma) = 0$, which implies $0<\kappa(\theta) \le \widetilde{\kappa}({\theta})$, $\forall{\theta}>\gamma$, therefore, we must have $0 < \widetilde{\gamma} \le \gamma$. Since $\kappa^{-S}$ is decreasing, we get $\kappa^{-S}(\gamma) \le \widetilde{\kappa}^{-S}(\widetilde{\gamma}) < 0$.
The other results follow analogically.
\end{proof}

\begin{remark}
The convex ordering is a conservatively sufficient condition for the ordering of the asymptotic decay rate, because it requires that the whole function lies above the other and it is not necessary for the asymptotic decay rate ordering.
\end{remark}

The following theorem shows that the manipulation of the dependence in the arrival process or the service process has an impact on the increment process or negative increment process of the queue. This is the mathematical foundation for dependence control through the arrival process or service process.

\begin{theorem}\label{theorem-ordering-arrival-or-service}
Let $\mathfrak{S}(t)=A(t)-S(t)$ and $\widetilde{\mathfrak{S}}(t)= \widetilde{A}(t) - \widetilde{S}(t)$.
If $\mathfrak{S}(t)$ and $\widetilde{\mathfrak{S}}(t)$ have identical service process $S(t)$, then
\begin{eqnarray}
A(t) \le_{cx} \widetilde{A}(t) &\implies& \mathfrak{S}(t) \le_{cx} \widetilde{\mathfrak{S}}(t).
\end{eqnarray}
If $\mathfrak{S}(t)$ and $\widetilde{\mathfrak{S}}(t)$ have identical arrival process $A(t)$, then
\begin{eqnarray}
S(t) \le_{cx} \widetilde{S}(t) &\implies& -\mathfrak{S}(t) \le_{cx} -\widetilde{\mathfrak{S}}(t).
\end{eqnarray}
\end{theorem}

\begin{proof}
Consider a convex function $f$. 
Let $X=S(t)$, $Y=\widetilde{S}(t)$, and $Z=A(t)$, $\forall{t}$.
Let $g(z)= \mathbb{E}[f(X-z)]$ and $h(z)= \mathbb{E}[f(Y-z)]$. As the function $x \mapsto f(x-z)$ is convex for all $z\in\mathbb{R}$, $X\le_{cx}Y$ implies $g(z)\le h(z)$ for all $z\in\mathbb{R}$. Thus
\begin{equation}
\mathbb{E}[f(X-Z)] = \mathbb{E}[g(Z)] \le \mathbb{E}[h(Z)] = \mathbb{E}[f(Y-Z)].
\end{equation}
The results follows directly.
\end{proof}

The following results show that the both the arrival and the service have an impact on the system performance, in other words, the stochastic dependence in both arrival and service can be taken advantage of for performance improvement.

\begin{corollary}
Let $\mathfrak{S}(t)=A(t)-S(t)$ and $\widetilde{\mathfrak{S}}(t)= \widetilde{A}(t) - \widetilde{S}(t)$.
If $\mathfrak{S}(t)$ and $\widetilde{\mathfrak{S}}(t)$ have identical service process $S(t)$, then
\begin{IEEEeqnarray}{rCl}
A(t) \le_{cx} \widetilde{A}(t),\ \forall{t}\in\mathbb{N} 
&\implies& 0 < \widetilde{\gamma} \le \gamma  \\
&\implies& \kappa^{-S}(\gamma) \le \widetilde{\kappa}^{-S}(\widetilde{\gamma}) < 0.
\end{IEEEeqnarray}
If $\mathfrak{S}(t)$ and $\widetilde{\mathfrak{S}}(t)$ have identical arrival process $A(t)$, then
\begin{IEEEeqnarray}{rCl}
S(t) \le_{cx} \widetilde{S}(t),\ \forall{t}\in\mathbb{N} 
&\implies& 0 > \widetilde{\gamma} \ge \gamma \\
&\implies&  \kappa^{-A}(\gamma) \ge \widetilde{\kappa}^{-A}(\widetilde{\gamma}) > 0.
\end{IEEEeqnarray}
\end{corollary}

\begin{proof}
The proof follows directly from Theorem \ref{theorem-ordering-arrival-and-service} and \ref{theorem-ordering-arrival-or-service}.
\end{proof}

The following result complements that the random variables with comonotonicity have the maximum summation in convex order \cite{dhaene2002concept}.

\begin{theorem}
Given an arbitrary sequence of random variables, $X_1, X_2, \ldots, X_n$, the mean of these random variables has the minimum summation in convex order, i.e.,
\begin{equation}
\sum_{i=1}^{n} \mathbb{E}\left[ X_i \right] \le_{cx} \sum_{i=1}^{n} X_i.
\end{equation}
\end{theorem}

\begin{proof}
According to Jensen's inequality,
\begin{equation}
\phi\left( \mathbb{E}\left[ \sum_{i=1}^{n} X_i \right]  \right) \le \mathbb{E}\left[ \phi\left( \sum_{i=1}^{n} X_i  \right) \right],\ \forall \phi, 
\end{equation}
where $\phi$ is a convex function.
\end{proof}

The application of the above result is the following corollary, particularly, it is an extension of the backlog result for stationary scenario in \cite{glynn1994logarithmic}.

\begin{corollary}\label{corollary-constant-optimal}
Under the conditions in Proposition \ref{proposition-arrival-service}.
For stationary or non-stationary arrival and service processes, if the mean of the time average exist, i.e.,
\begin{equation}
\mathbb{E}[X] = \lim_{n\rightarrow\infty} \sum_{i=1}^{n} \mathbb{E}\left[ X_i \right],
\end{equation}
then the situation with the constant arrival process with the same mean or constant service process with the same mean has the largest asymptotic decay rate of delay and backlog.
\end{corollary}

\subsubsection{More than Independence}

We classify the dependence into three types, i.e., positive dependence, independence, and negative dependence.
Intuitively, positive dependence implies that large or small values of random variables tend to occur together, while negative dependence implies that large values of one variable tend to occur together with small values of others \cite{denuit2006actuarial}.
We illustrate the dependence classification in terms of the service process, i.e., wireless channel capacity, similar arguments hold analogically for the arrival process.

\begin{proposition}
The capacity $\bm{C} = \left( C(1), \ldots, C(t) \right)$ is said to have a positive dependence structure $\bf{C}_{+}$ in the sense of supermodular order, if
\begin{equation}
\bf{C}_{\perp} \le_{sm} \bf{C}_{+},
\end{equation}
or a negative dependence structure $\bf{C}_{-}$ in the sense of supermodular order, if
\begin{equation}
\bf{C}_{-} \le_{sm}\bf{C}_{\perp},
\end{equation}
where $\bf{C}_{\perp}$ has an independence structure. 
\end{proposition}

The convex order gives a sufficient condition for the ordering of the asymptotic decay rate of delay and the following result relates the multivariate supermodular order to the univariate convex order.

\begin{lemma}
The supermodular ordering of instantaneous capacity,
$
\textbf{C} \le_{sm} \widetilde{\textbf{C}},
$
entails that the marginal distributions of the instantaneous increments are identical, 
particularly, \cite{sun2017statistical}
\begin{equation}
\bm{C}\le_{sm}\widetilde{\bm{C}} \implies \sum_{i=1}^{t}C(i)\le_{cx}\sum_{i=1}^{t}\widetilde{C}(i).
\end{equation}
\end{lemma}

The following result shows the direct relationship between dependence and the asymptotic decay rate.

\begin{corollary}
Consider an identical arrival process and two wireless channel capacity processes, if the capacities are supermodular ordered, then the asymptotic decay rates are correspondingly ordered, i.e.,
\begin{eqnarray}
\textbf{C} \le_{sm} \widetilde{\textbf{C}},\ \forall{t}\in\mathbb{N} 
&\implies&
0 > \widetilde{\gamma} \ge \gamma, \\
\textbf{C} \le_{sm} \widetilde{\textbf{C}},\ \forall{t}\in\mathbb{N}
&\implies&
\kappa^{-A} ( {\gamma} ) \ge  \widetilde{\kappa}^{-A} (\widetilde{\gamma}) > 0.
\end{eqnarray}
\end{corollary} 

\begin{remark}
The wireless channel performance relies on these statistical properties of the capacity process that are more than the mean of the instantaneous capacity, for instance, stochastic dependence.
The above results show that stochastic dependence provides another degree of freedom for improving wireless channel performance, which is especially crucial in the extreme scenario where there is no more gain in terms of the capacity mean.
\end{remark}

\subsection{Transform of Dependence}

In this subsection, we classify the random parameters in wireless channel capacity into controllable and uncontrollable random parameters and show how to transform the dependence structure of the capacity by manipulating the controllable random parameters of the capacity, and the results for a single channel are extended to complex channels composing of a set of sub-channels. 
In addition, since the dependence in the arrival process also influences the channel performance, the dependence manipulation in deterministic multiplexing and random multiplexing are also discussed.

\subsubsection{Manipulation of Service}

We treat the capacity $C(t)$ as a function $f_t$ of a set of random parameters $\left( X_t^1, X_t^2, \ldots, X_t^n \right)$, 
and we specify that the function is time-variant, the dimension of the set is time-invariant, and the function $f_t$ is increasing or decreasing at $X_t^i$ for all the time, i.e.,
\begin{equation}
C(t) = f_{t} \left( X_t^1, X_t^2, \ldots, X_t^n \right).
\end{equation}
In other words, we treat the capacity as a functional of a multivariate stochastic process and the functional maps the multivariate stochastic process to a univariate stochastic process.

The following theorem shows that the manipulation of the dependence in the controllable parameter processes transforms the dependence structure of the capacity.

\begin{theorem}\label{theorem-dependence-transform}
Assume the random parameters are spatially independent and temporally dependent.
The supermodular ordering of a parameter series implies the ordering of the capacity, i.e., for any $1\le i\le n$,
\begin{multline}
\left( X_1^i, X_2^i, \ldots, X_t^i \right) \le_{sm} \left( \widetilde{X}_1^i, \widetilde{X}_2^i, \ldots, \widetilde{X}_t^i \right) \\
\implies \left( C(1), C(2), \ldots, C(t) \right) \le_{sm} \left( \widetilde{C}(1), \widetilde{C}(2), \ldots, \widetilde{C}(t) \right),
\end{multline}
where $\widetilde{C}(j) = f_j\left( X_j^1, \ldots, X_j^{i-1}, \widetilde{X}_j^i, X_j^{\wedge(i+1,n)}, \ldots, X_j^n \right)$, $\forall 1\le j\le t$.
\end{theorem}

\begin{proof}
Without loss of generalization, we consider the supermodular order of the random parameters with index $1$ in the proof.

For all increasing or all decreasing functions $g_i: \mathbb{R}\rightarrow\mathbb{R}$, $i=1, \ldots, t$, 
\begin{equation}
\left( X_1^1, X_2^1, \ldots, X_t^1 \right) \le_{sm} \left( \widetilde{X}_1^1, \widetilde{X}_2^1, \ldots, \widetilde{X}_t^1 \right),
\end{equation}
implies
\begin{equation}
\left( g_1\left( X_1^1 \right), \ldots, g_t\left( X_t^1 \right) \right) \le_{sm} \left( g_1\left( \widetilde{X}_1^1 \right), \ldots, g_t\left( \widetilde{X}_t^1 \right) \right),
\end{equation}
because a composition of a supermodular function with coordinatewise functions, that are all increasing or are all decreasing, is a supermodular function \cite{shaked2007stochastic}.

Let $\bm{Z}_t = \left( X^2_t, X^3_t, \ldots, X^n_t \right)$, assume $\bm{Z}_t$ is independent of $X^1_t$, $\forall t$, $\left( g_1\left( X_1^1 \right), \ldots, g_t\left( X_t^1 \right) \right) \le_{sm} \left( g_1\left( \widetilde{X}_1^1 \right), \ldots, g_t\left( \widetilde{X}_t^1 \right) \right)$ implies
\begin{multline}
\left( f_1\left( X_1^1, \bm{Z}_1  \right), \ldots, f_t\left( X_t^1, \bm{Z}_t \right) \Big| \left( \bm{Z}_1, \ldots, \bm{Z}_t  \right) = \bm{z} \right) \\
\le_{sm} \left( f_1\left( \widetilde{X}_1^1, \bm{Z}_1  \right), \ldots, f_t\left( \widetilde{X}_t^1, \bm{Z}_t \right) \Big| \left( \bm{Z}_1, \ldots, \bm{Z}_t  \right) = \bm{z} \right), \forall \bm{z},
\end{multline}
whenever $f_i\left( x_i^1, \bm{z}_i  \right)$, $\forall i$, are all increasing or are all decreasing in $x_i^1$ for every $\bm{z}_i$, 
and it further implies
\begin{multline}
\left( f_1\left( X_1^1, \bm{Z}_1  \right), \ldots, f_t\left( X_t^1, \bm{Z}_t \right)  \right) \\
\le_{sm} \left( f_1\left( \widetilde{X}_1^1, \bm{Z}_1  \right), \ldots, f_t\left( \widetilde{X}_t^1, \bm{Z}_t \right)  \right),
\end{multline}
because the sumermodular order is closed under mixtures \cite{shaked2007stochastic}.
\end{proof}

The following theorem shows that a greater number of controllable random parameter processes brings a  stronger transform to the dependence structure of capacity.

\begin{theorem}\label{theorem-transform-strength}
Assume the random parameters are spatially independent and temporally dependent.
Consider there are $i$, $1\le i\le n$, controllable random parameters. If
\begin{eqnarray}
\left( X_1^j, X_2^j, \ldots, X_t^j \right) \le_{sm} \left( \widetilde{X}_1^j, \widetilde{X}_2^j, \ldots, \widetilde{X}_t^j \right),\ \forall 1\le j \le i,
\end{eqnarray}
then
\begin{equation}
\widetilde{\bm{C}}_{t}^{k} \le_{sm} \widetilde{\bm{C}}_{t}^{j},\ \forall 0\le k \le j \le i, 
\end{equation}
where
$
\widetilde{{C}}_{t_m}^{k} = f_m\left( \widetilde{X}_m^1, \ldots, \widetilde{X}_m^k, {X}_m^{\wedge (k+1, n)}, \ldots, {X}_m^n   \right),\ 1\le m \le t.
$
\end{theorem}

\begin{proof}
According to Theorem \ref{theorem-dependence-transform}, 
\begin{IEEEeqnarray}{rCl}
\IEEEeqnarraymulticol{3}{l}{
\left( f_1\left( X_1^1, X_1^2, \ldots, X_1^n \right), \ldots, f_t\left( X_t^1, X_t^2, \ldots, X_t^n \right) \right) }\\
&\le_{sm}& \left( f_1\left( \widetilde{X}_1^1, X_1^2, \ldots, X_1^n \right), \ldots, f_t\left( \widetilde{X}_t^1, X_t^2, \ldots, X_t^n \right) \right) \\
&\le_{sm}& \left( f_1\left( \widetilde{X}_1^1, \widetilde{X}_1^2, X_1^3, \ldots, X_1^n \right), \ldots, f_t\left( \widetilde{X}_t^1, \widetilde{X}_t^2, X_t^3, \ldots, X_t^n \right) \right), \IEEEeqnarraynumspace
\end{IEEEeqnarray}
and the result follows iteratively because of the transitivity property of supermodular order \cite{muller2002comparison}.
\end{proof}

The following theorem shows that the manipulation of dependence in a sub-channel capacity transforms the dependence structure of the overall channel capacity, the more number of manipulated sub-channels the stronger dependence transform strength on the overall capacity.

\begin{theorem}\label{theorem-subchannel-control}
Consider a wireless channel composing of  $M$ independent sub-channels, the instantaneous capacity of the overall channel is a function of the instantaneous capacity of each sub-channels, i.e., 
\begin{equation}
{\mathfrak{C}}_t = f_t\left(C_t^1, \ldots, {C}_t^{{M}}  \right).
\end{equation}
For example, the overall capacity is the summation of the capacity of each sub-channel, i.e.,
\begin{equation}
\left(C_t^1, \ldots, {C}_t^{{M}}  \right) \mapsto f_t\left(C_t^1, \ldots, {C}_t^{{M}}  \right) = \sum_{m=1}^{{M}} {C}_t^m.
\end{equation}
Assume the function is always increasing or decreasing at the instantaneous capacity of each sub-channels. 

Then, for any $1\le i\le {M}$,
\begin{multline}
\left( C_1^i, \ldots, C_t^i \right) \le_{sm} \left( \widetilde{C}_1^i, \ldots, \widetilde{C}_t^i \right) \\
\implies
\left( \mathfrak{C}_1, \ldots, \mathfrak{C}_t \right) \le_{sm} ( \widetilde{\mathfrak{C}}_1, \ldots, \widetilde{\mathfrak{C}}_t ),
\end{multline}
where $\widetilde{\mathfrak{C}}_j = f_j\left(C_j^1, \ldots, \widetilde{C}_j^i, {C}_j^{\wedge({i+1,M})}, \ldots, {C}_j^{{M}}  \right)$, $1\le j\le t$.

In addition, if
\begin{eqnarray}
\left( C_1^i, \ldots, C_t^i \right) \le_{sm} \left( \widetilde{C}_1^i, \ldots, \widetilde{C}_t^i \right),\ \forall 1\le i \le M,
\end{eqnarray}
then
\begin{equation}
\widetilde{\bm{\mathfrak{C}}}_{t}^{k} \le_{sm} \widetilde{\bm{\mathfrak{C}}}_{t}^{j},\ \forall 0\le k \le j \le M, 
\end{equation}
where
$
\widetilde{\mathfrak{C}}_{t_m}^{k} = f_m\left( \widetilde{C}_m^1, \ldots, \widetilde{C}_m^k, {C}_m^{\wedge (k+1, M)}, \ldots, {C}_m^M   \right),\ 1\le m \le t.
$
\end{theorem}

\begin{proof}
Considering Theorem \ref{theorem-dependence-transform} and \ref{theorem-transform-strength}, the results are obvious.
\end{proof}

\subsubsection{Manipulation of Arrival}

We show the impact of the dependence manipulation of the individual processes on the aggregated process in deterministic multiplexing and random multiplexing.

The following theorem shows how to transform the dependence structure of the arrival process in the deterministic multiplexing.

\begin{theorem}
For deterministic multiplexing of a set of arrival processes, the dependence manipulation in an individual arrival process transforms the dependence structure of the aggregate process,
the more number of manipulated processes the stronger strength of dependence transform.
\end{theorem}

\begin{proof}
Considering Theorem \ref{theorem-subchannel-control}, the argument is obvious.
\end{proof}

For random multiplexing, the dependence in the random multiplexing control has an impact on the dependence structure of the aggregated process. Particularly, the dependence manipulation in an individual process brings its impact into the randomly multiplexed process.
We prove the theorem in Appendix \ref{proof-of-theorem-random-multiplexing}.

\begin{theorem}\label{theorem-random-multiplexing}
Let $\bm{X}_j = ( X_{j,1}, \ldots, X_{j,m} )$ and $\bm{Y}_j = ( Y_{j,1}, \ldots,  Y_{j,m} )$, $j = 1,2,\ldots$, be two independent sequences of non-negative random vectors, and let $\bm{M} = \left(M_1, M_2, \ldots , M_m \right)$ and $\bm{N} = \left( N_1, N_2, \ldots, N_m \right)$ be two vectors of non-negative integer-valued random variables. Assume that both $\bm{M}$ and $\bm{N}$ are independent of the $X_j$'s and $Y_j$'s. 

If $\bm{M} \le_{sm} \bm{N}$, then
\begin{equation}
\left( \sum_{j=1}^{M_1} X_{j,1}, \ldots, \sum_{j=1}^{M_m} X_{j,m}  \right) \le_{sm} \left( \sum_{j=1}^{N_1} X_{j,1}, \ldots, \sum_{j=1}^{N_m} X_{j,m}  \right).
\end{equation}
If $\bm{X}_j  \le_{sm} \bm{Y}_j $, $\forall j$, then
\begin{equation}
\left( \sum_{j=1}^{N_1} X_{j,1}, \ldots, \sum_{j=1}^{N_m} X_{j,m}  \right) 
 \le_{sm} \left( \sum_{j=1}^{N_1} Y_{j,1}, \ldots, \sum_{j=1}^{N_m} Y_{j,m}  \right).
\end{equation}
If $\bm{M} \le_{sm} \bm{N}$ and $\bm{X}_j  \le_{sm} \bm{Y}_j $, $\forall j$, then
\begin{equation}
\left( \sum_{j=1}^{M_1} X_{j,1}, \ldots, \sum_{j=1}^{M_m} X_{j,m}  \right) \le_{sm}
\left( \sum_{j=1}^{N_1} Y_{j,1}, \ldots, \sum_{j=1}^{N_m} Y_{j,m}  \right).
\end{equation}
\end{theorem}

\begin{remark}
The supermodular ordering of the random vectors is a sufficient condition for the convex ordering of the partial sum of the random vector, but it is not a necessary condition, e.g., an increasing in the mean of one of the random variables results in the convex ordering, particularly, the directional convex order is fit for investigating the impact of the marginals.
In other words, the sufficient and necessary condition for the asymptotic decay rate ordering is more than dependence, while dependence ordering is the focus of this paper.
\end{remark}

\section{Application}
\label{application}

This section focuses on the application of dependence control to wireless channel capacity with a specific dependence structure, i.e., the cumulative capacity is modeled as a Markov additive process. 
Performance measures for arrival processes as Markov additive process and constant process are discussed.
Copula is used to represent and manipulate the dependence structure of the underlying Markov process.

\subsection{Structure Specification}

We model the wireless channel capacity as a Markov additive process and the specification is as follows.

\begin{proposition}
If the dependence in capacity is driven by a Markov process and the incremental capacity has a specific distribution with respect to a specific state transition, then the additive capacity together with the underlying Markov process 
form a Markov additive process. 
\end{proposition}

We use the definition of Markov additive process in \cite{asmussen2003applied}\cite{asmussen2010ruin}, and 
we focus on the finite state space scenario in discrete-time setting, where the structure is fully understood.

A Markov additive process is defined as a bivariate Markov process $\{X_t\}=\{(J_t,S(t))\}$ where $\{J_t\}$ is a Markov process with state space $E$ and the increments of $\{S(t)\}$ are governed by $\{J_t\}$ in the sense that 
\begin{equation}
\mathbb{E}\left[f(S({t+s})-S(t))g(J_{t+s})|\mathscr{F}_t\right] = \mathbb{E}_{J_t,0}\left[f(S(s))g(J_s)\right].
\end{equation}
In discrete time, a Markov additive process is specified by the measure-valued matrix (kernel) $\mathbf{F}(dx)$ whose $ij$th element is the defective probability distribution 
\begin{equation}
F_{ij}(dx)= \mathbb{P}_{i,0}(J_1=j,Y_1\in{dx}),
\end{equation}
where $Y_t=S(t)-S(t-1)$. An alternative description is in terms of the transition matrix $\mathbf{P}=(p_{ij})_{i,j\in{E}}$, $p_{ij}= \mathbb{P}_i(J_1=j)$, and the probability measures
\begin{equation}
H_{ij}(dx) = \mathbb{P}(Y_1\in{dx}|J_0=i, J_1=j) = \frac{F_{ij}(dx)}{p_{ij}}.
\end{equation} 
With respect to a transition probability $p_{ij}$, the increment of $\{S_t\}$ has a distribution $B_{ij}$.

Consider the matrix $\widehat{\textbf{F}}_t[\theta]=( \mathbb{E}_i[e^{\theta{S(t)}};J_t=j])_{i,j\in{E}}$. 
In discrete time,
\begin{equation}
\widehat{\textbf{F}}_t[\theta]=\widehat{\textbf{F}}[\theta]^t,
\end{equation}
where $\widehat{\textbf{F}}[\theta]=\widehat{\textbf{F}}_1[\theta]$ is a $E\times{E}$ matrix with $ij$th element $\widehat{F}^{(ij)}[\theta]=p_{ij}\int{e^{\theta{x}}F^{(ij)}}(dx)$, and $\theta\in\Theta=\{ \theta\in{R}:\int{e^{\theta{x}}F^{(ij)}}(dx)<\infty \}$  \cite{asmussen2003applied}. 
By Perron-Frobenius theorem, $\widehat{\textbf{F}}[\theta]$ has a positive real eigenvalue with maximal absolute value, $e^{\kappa(\theta)}$, in discrete time.
The corresponding right and left eigenvectors are respectively $\textbf{h}{(\theta)}=\left(h_{i}{(\theta)}\right)_{i\in{E}}$ and $\textbf{v}{(\theta)}=\left(v_{i}{(\theta)}\right)_{i\in{E}}$, particularly, $\textbf{v}{(\theta)}$, $\textbf{v}{(\theta)}\textbf{h}{(\theta)}=1$ and $\bm{\pi}\textbf{h}{(\theta)}=1$, where $\bm{\pi}=\textbf{v}{(0)}$ is the stationary distribution and $\textbf{h}{(0)}=\textbf{e}$. 
Particularly, the likelihood ratio process in the exponential change of measure is of interest \cite{asmussen2003applied},
\begin{equation}
L(t) = \frac{h_{J_t}{(\theta)}}{h_{J_0}{(\theta)}}e^{\theta{S(t)}-t\kappa(\theta)},
\end{equation}
which is a mean-one martingale. This martingale process is useful for wireless channel performance analysis.

\subsection{Performance Analysis}
 
The following theorem presents the results of the delay and backlog tail probability on infinite time horizon, for Markov additive arrival process and Markov additive service process.
We present the proof in Appendix \ref{proof-of-theorem-performance-bound}.

\begin{theorem}\label{theorem_performance_bound}
Consider a Markov additive arrival $A(t)$ with state space $E$ and initial state distribution ${\bm\varpi}^A_0$, and a Markov additive capacity $S(t)$ with state space $E'$ and initial distribution ${\bm\varpi}^S_0$. Specifically, given the initial state distribution, the state distribution at time $t$ is ${\bm\varpi}_t={\bm\varpi}_0{\bm P}^t$. Assume independence between the arrival process and the service process. The delay tail probability, conditional on the initial state ${\bm{J}}_{d,0}=\bm{i}$, i.e., $\left\{ J^{A}_d, J^{-S}_0 \right\}=\left\{ i^A, i^{-S} \right\}$, is expressed as
\begin{equation}
H_{-}^D \cdot h_{J_0}^{-S}(\theta) \cdot e^{-d\kappa(\theta)} \le \mathbb{P}_{\bm i} ( D > d ) \le H_{+}^D \cdot h_{J_0}^{-S}(\theta)  \cdot e^{-d\kappa(\theta)},
\end{equation}
where 
\begin{IEEEeqnarray}{rCl}
H_{+}^D &=& \frac{\max_{j\in{E}}h_{j}^{A}{(\theta)}}{\min_{j\in{E}}h_{j}^{A}{(\theta)}} \cdot \frac{1}{\min_{j\in{E'}}h_{j}^{-S}(\theta)}, \\
H_{-}^D &=& e^{-\kappa^{A}(\theta)} \cdot \left( \frac{\min_{j\in{E}}h_{j}^{A}{(\theta)}}{\max_{j\in{E}}h_{j}^{A}{(\theta)}} \right)^2 \cdot \frac{1}{\max_{j\in{E'}}h_{j}^{-S}(\theta)}. \IEEEeqnarraynumspace
\end{IEEEeqnarray}
The backlog tail probability, conditional on the initial state ${\bm{J}}_{0,0}=\bm{i}$, i.e., $\left\{ J^{A}_0, J^{-S}_0 \right\}=\left\{ i^A, i^{-S} \right\}$, is expressed as
\begin{equation}
H_{-}^B \cdot {h_{J_0}^{A}{(\theta)}}{h_{J_0}^{-S}(\theta)} \cdot e^{-\theta{b}} \le \mathbb{P}_{\bm i} ( B > b ) \le H_{+}^B \cdot {h_{J_0}^{A}{(\theta)}}{h_{J_0}^{-S}(\theta)} \cdot e^{-\theta{b}},
\end{equation}
where 
\begin{IEEEeqnarray}{rCl}
H_{+}^B &=& \frac{ 1 }{\min_{j\in{E}}h_{j}^{A}{(\theta)}} \cdot \frac{1}{\min_{j\in{E'}}h_{j}^{-S}(\theta)}, \\
H_{-}^B &=& { e^{-\kappa^{A}(\theta)} } \cdot {\min\limits_{j\in{E}} h_{j}^{A}(\theta) } \cdot \frac{1}{ \left( \max_{j\in{E}}h_{j}^{A}{(\theta)} \right)^2} \cdot \frac{1}{\max\limits_{j\in{E'}}h_{j}^{-S}(\theta)}. \IEEEeqnarraynumspace
\end{IEEEeqnarray}
For the delay and backlog, $\theta$ is the root of the stability equation, i.e.,
\begin{eqnarray}
\theta = \left\{ \theta>0: \kappa^{A}(\theta) + \kappa^{-S}(\theta) = 0 \right\},
\end{eqnarray}
and $\kappa(\theta): = \kappa^{A}(\theta)=-\kappa^{-S}(\theta)$.
\end{theorem}

Note the delay and backlog results are non-asymptotic, and these non-asymptotic results have an identical decay rate as the asymptotic decay rates, because they use a stability condition in the asymptotic regime.

The following theorem presents the time-dependent delay and backlog tail probability on finite time horizon, which is a function of the violated delay and backlog, and the decay rates are time-variant. We prove the theorem in Appendix \ref{proof-of-theorem-timely-performance}.

\begin{theorem}\label{theorem-timely-performance}
Consider the same specification as in Theorem \ref{theorem_performance_bound}.
For delay, let $\gamma$ be the root to $\kappa^{A}(\theta) + \kappa^{-S}(\theta)=0$,
$y_\gamma = \frac{{\dot{\kappa}}^{A}(\gamma)}{{\dot\kappa}^{A}(\gamma) + {\dot\kappa}^{-S}(\gamma)}$; given any fixed $y>1$, $\theta$ is the root to $y{\dot\kappa}^{-S}(\theta)=-(y-1){\dot\kappa}^{A}(\theta)$, and $\theta_y = -y\kappa^{-S}(\theta) -(y-1)\kappa^{A}(\theta)$,
then
\begin{IEEEeqnarray}{rCl}
\mathbb{P}_{\bm i}(D(t) >d; t \le yd) &\le& H_{+}^{D} h_{J_0}^{-S}(\theta) e^{-d \theta_y},\ y< y_{\gamma}, \\
\mathbb{P}_{\bm i}(D >d) - P_{\bm i}(D(t) >d; t\le yd) &\le& H_{+}^{D} h_{J_0}^{-S}(\theta) e^{-d \theta_y},\ y> y_\gamma, \IEEEeqnarraynumspace
\end{IEEEeqnarray}
where $H_{+}^{D} = \frac{\max_{j\in{E}}h_{j}^{A}{(\theta)}}{\min_{j\in{E}}h_{j}^{A}{(\theta)}} \cdot \frac{1}{\min_{j\in{E'}}h_{j}^{-S}(\theta)}$.

For backlog, let $\gamma$ be the root to $\kappa^{A}(\theta) + \kappa^{-S}(\theta)=0$, $y_\gamma = \frac{1}{{\dot\kappa}^{A}(\gamma) + {\dot\kappa}^{-S}(\gamma)}$;
given any fixed $y>0$, $\theta$ is the root to $y\left( \dot\kappa^{A}(\theta) + \dot\kappa^{-S}(\theta) \right) =1$, and $\theta_y = \theta - y\left( \kappa^{A}(\theta) + \kappa^{-S}(\theta) \right)$, then
\begin{IEEEeqnarray}{rCl}
\mathbb{P}_{\bm i}(B(t) >b; t \le yb) &\le& H_{+}^{B} h_{J_0}^{A}(\theta) h_{J_0}^{-S}(\theta) e^{-b \theta_y},\ y< y_{\gamma}, \\
\mathbb{P}_{\bm i}(B >b) - \mathbb{P}_{\bm i}(B(t) >b; t\le yb) &\le& H_{+}^{B} h_{J_0}^{A}(\theta) h_{J_0}^{-S}(\theta) e^{-b \theta_y},\ y> y_\gamma, \IEEEeqnarraynumspace
\end{IEEEeqnarray}
where $H_{+}^{B} = \frac{1}{\min_{j\in{E}}h_{j}^{A}{(\theta)}} \cdot \frac{1}{\min_{j\in{E'}}h_{j}^{-S}(\theta)}$.
\end{theorem}

Note for the above infinite time and finite time results, we only give results of the conditional tail probability, which is sufficient to calculate the tail probability by averaging over the initial state.

The following result shows an upper bound of the time average of the Markov additive process.

\begin{lemma}\label{lemma_arrival_mean}
Consider a Markov additive process, $M(t)$, the mean of  time average of the process is smaller and equal than $\frac{\kappa(\theta)}{\theta}$, no matter the initial state, i.e.,
\begin{equation}
\lim_{t\rightarrow{\infty}} \mathbb{E}\left[ \frac{M(t)}{t} \right] \le \frac{\kappa(\theta)}{\theta},\ \forall \theta>0.
\end{equation}
\end{lemma}

\begin{proof}
For the Markov additive process $M(t)$, it's shown \cite{asmussen2010ruin}
\begin{equation}
\mathbb{E}_i\left[ e^{\theta M(t)}h_{J_t}^{(\theta)} \right] = h_i^{(\theta)}e^{t\kappa(\theta)}.
\end{equation}
With a logarithm operation, it follows Jensen's inequality that
\begin{equation}
\mathbb{E}_i\left[ \theta M(t) \right] \le t\kappa(\theta) + \log{h_i^{(\theta)}} - \mathbb{E}_i\left[ \log h_{J_t}^{(\theta)} \right],
\end{equation}
calculate the time average and let time go to infinity, $\forall \theta>0$,
\begin{equation}
\lim_{t\rightarrow{\infty}} \mathbb{E}_i\left[ \frac{M(t)}{t} \right] \le \frac{\kappa(\theta)}{\theta},
\end{equation}
where the right hand side is independent of the initial state.
\end{proof}

An upper bound of the delay-constrained capacity is as follows.

\begin{corollary}
Consider a Markov additive arrival $A(t)$ with state space $E$ and initial state distribution ${\bm\varpi}^A_0$, and a Markov additive capacity $S(t)$ with state space $E'$ and initial distribution ${\bm\varpi}^S_0$. Specifically, given the initial state distribution, the state distribution at time $t$ is ${\bm\varpi}_t={\bm\varpi}_0{\bm P}^t$. The delay-constrained capacity, conditional on the initial state ${\bm{J}}_{d,0}=\bm{i}$, i.e., $\left\{ J^{A}_d, J^{-S}_0 \right\}=\left\{ i^A, i^{-S} \right\}$, is bounded by
\begin{eqnarray}
\overline{C}_{\bm{i}}(d,\epsilon) \le \frac{-1}{\theta d}\log\frac{\mathbb{P}_{\bm i}(D > d)}{H^{D}_{+}\cdot h^{-S}_{J_0}(\theta)},
\end{eqnarray}
where
$H_{+}^D = \frac{\max_{j\in{E}}h_{j}^{A}{(\theta)}}{\min_{j\in{E}}h_{j}^{A}{(\theta)}} \cdot \frac{1}{\min_{j\in{E'}}h_{j}^{-S}(\theta)}$, 
and the delay-constrained capacity is bounded by
\begin{eqnarray}
\overline{C}(d,\epsilon) \le \sum {\bm\varpi}_{\bm i}  \frac{-1}{\theta d}\log\frac{\mathbb{P}_{\bm i}(D > d)}{H^{D}_{+}\cdot h^{-S}_{i}(\theta)},
\end{eqnarray}
where ${\bm\varpi}_{\bm i} = {\bm\varpi}^A_d \times {\bm\varpi}^{-S}_0$. The parameter $\theta$ is free for optimization.
\end{corollary}

\begin{proof}
According to Lemma \ref{lemma_arrival_mean},
the delay-constrained capacity is bounded by
\begin{equation}
\overline{C}_{\bm{i}}(d,\epsilon) \le \sup_{ \mathbb{P}_{\bm{i}}(D>d)\le\epsilon } \frac{\kappa^{A}(\theta)}{\theta},
\end{equation}
then the result directly follows Theorem \ref{theorem-delay-constrained-capacity} and \ref{theorem_performance_bound}.
\end{proof}

\subsubsection{Peak Performance}

According to Corollary \ref{corollary-constant-optimal}, the wireless channel attains the best performance for constant arrival process in terms of the asymptotic decay identities, thus the constant arrival is fit for investigating the ultimate quality of service that the wireless channel can provide. On the other hand, it indicates that the ultimate wireless channel performance is solely determined by the statistical properties of the wireless channel regardless of the arrivals, in terms of some measure identities.

Assume the input is a constant process, 
\begin{equation}
A(t) = \lambda{t},
\end{equation}
Theorem \ref{theorem_performance_bound} reduces to a special case, of which the results are available in \cite{sun2017statistical} and restated here.

\begin{corollary}\label{corollary_performance_bound}
Consider a constant arrival process $A(t)=\lambda{t}$, the delay conditional on the initial state $J_0=i$ is bounded by
\begin{equation}
e^{-\theta A(1)} \cdot \frac{h_{J_0}{(-\theta)} e^{-\theta{\lambda{d}}} }{\max\limits_{j\in{E'}}h_{j}{(-\theta)}} \le \mathbb{P}_i(D > {d}) \le \frac{h_{J_0}{(-\theta)} e^{-\theta{\lambda{d}}} }{\min\limits_{j\in{E'}}h_{j}{(-\theta)}},
\end{equation}
where $-\theta$ is the negative root of $\kappa(\theta)=0$ of $S(t)-\lambda{t}$ and $\bm{h}{(-\theta)}$ is the corresponding right eigenvector,
given the initial state distribution $\bm{\varpi}$, the delay and backlog are bounded by 
\begin{eqnarray}
\mathbb{P}(D > d) &=& \sum_{i\in{E'}}\varpi_{i} \mathbb{P}_i(D > d), \\
\mathbb{P}(B > b) &=& \mathbb{P}(D > b/\lambda). 
\end{eqnarray}
\end{corollary}

The delay-constrained capacity for the constant arrival is shown as follows.

\begin{corollary}
For constant fluid traffic $A(t)=\lambda{t}$, the delay-constrained capacity, letting $\mathbb{P}(D > d)=\epsilon$, is expressed as
\begin{equation}
\frac{-1}{\theta{d}}{\log\frac{ e^{\theta A(1)} \cdot \epsilon\cdot{\max\limits_{j\in{E'}}h_{j}{(-\theta)}}}{\sum_{i\in{E'}}{\varpi_i}h_{i}{(-\theta)}}} \le \lambda \le \frac{-1}{\theta{d}}{\log\frac{\epsilon\cdot{\min\limits_{j\in{E'}}h_{j}{(-\theta)}}}{\sum_{i\in{E'}}{\varpi_i}h_{i}{(-\theta)}}}.
\end{equation}
\end{corollary}

\begin{proof}
Consider the delay-constrained capacity for the constant fluid process $A(t) =\lambda{t}$,
\begin{equation}
\overline{C}{(d,\epsilon)} = \sup_{\mathbb{P}(D(t) > d) \le \epsilon, \forall t} \lambda,
\end{equation}
the result follows directly from Corollary \ref{corollary_performance_bound}.
\end{proof}

The ordering of the asymptotic decay rate reflects a rough ordering of the tail distribution, specifically, if the asymptotic stability condition is used, the tail distribution bounds have an identical decay rate in finite time regime and infinite time regime.
The impact of negative dependence and positive dependence in capacity on delay and comparison with independence in capacity are illustrated in Fig. \ref{delay_markov_negative_positive_dependence}.
We fix the noise power density and change the transmission power in SNR. 
By introducing negative dependence through Fr\'{e}chet copula, the wireless channel attains a better performance with less transmission power or smaller capacity mean in contrast to introducing positive dependence.

\begin{figure}[!t]
\centering
\includegraphics[width=3.5in]{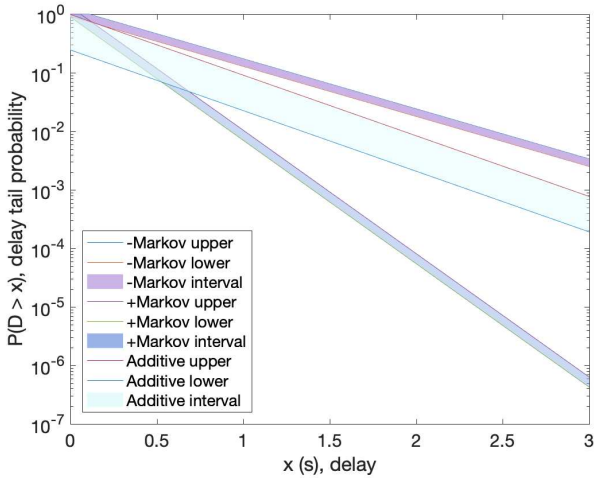}
\caption{Delay tail distribution of Rayleigh channel. ``-'' and ``+'' depict respectively negative and positive dependence in capacity, the lines depict the double-sided bounds with the intervals depicted as the shaded areas. $\lambda=10$kbits, $W=20$kHz, $SNR=e^{0.5}$ for the independence case of additive capacity process, $\textbf{SNR}=[e^{0.5}~e^{0.5}; 0.7e^{0.5}~0.7e^{0.5}]$, 
$\textbf{P}=[0.4125~0.5875; 0.2518~0.7482]$ 
calculated from Fr\'{e}chet copula with $\alpha=0.5$ for $\lambda-c(t)$ indicating positive dependence in capacity and 
$\textbf{P}=[0.2875~0.7125; 0.3054~0.6946]$ 
calculated from Fr\'{e}chet copula with $\alpha=-0.5$ for $\lambda-c(t)$ indicating negative dependence in capacity, 
for the dependence case of Markov additive capacity process with initial distribution $\bm{\varpi}=[0.3~0.7]$ and stationary distribution $\bm{\pi}=[0.3\ 0.7]$.
}
\label{delay_markov_negative_positive_dependence}
\end{figure}

\subsection{Dependence Manipulation}

We distinguish the random parameters in the wireless system, which cause the dependence in the wireless channel capacity, into two categories, i.e., uncontrollable parameters and controllable parameters.
Uncontrollable parameters represent the property of the environment that can not be interfered, e.g., fading, 
while controllable parameters represent the configurable property of the wireless system, e.g., power. 
We use the controllable parameters to induce negative dependence into the wireless channel capacity to achieve dependence control.

\subsubsection{Copula Representation}

Consider a joint distribution $F(X_1,$ $\ldots,X_n)$ with marginal distribution $F_i(x)$, $i=1,\ldots,n$. Denote $u_i = F_i\left(X_i\right)$, which is uniformly distributed in the unit interval, then \cite{embrechts2009copulas}
\begin{IEEEeqnarray}{rCl}
F\left( X_1,\ldots,X_n \right) &=& F\left( F_1^{-1}\left(u_1\right),\ldots, F_n^{-1}\left(u_n\right) \right) \\
&\equiv& C\left(u_1,\ldots,u_n\right),
\end{IEEEeqnarray}
where $C$ is a copula\footnote{In the wireless communication literature, $C$ usually represents capacity, while in the dependence modeling literature, $C$ usually represents copula, in this subsection, we abuse the notation to be consistent with the literature.} 
with standard uniform marginals, specifically, if the marginals are continuous, the copula is unique.

Let $(\Omega,\bm{\mathscr{F}},(\bm{\mathscr{F}}_t)_{t\in{N}}, \mathbb{P})$ be a filtered probability space and $(\bm{X}_t)_{t\in{N}}$ be an adapted stochastic process. 
$\bm{X}$ is a Markov process if and only if
\begin{IEEEeqnarray}{rCl}
\mathbb{P}\left( \bm{X}_t\le{x}|\bm{X}_{t-1}, \bm{X}_{t-2}, \ldots, \bm{X}_0 \right) = \mathbb{P}\left( \bm{X}_t\le{x}|\bm{X}_{t-1} \right). \IEEEeqnarraynumspace
\end{IEEEeqnarray}
The Markov property is solely a dependence property that can be modeled exclusively in terms of copulas \cite{darsow1992copulas,overbeck2015multivariate}.
Since there is no requirement on the $1$-dimensional marginal distribution $X_t^i$ for $\bm{X}$ to be Markov, starting with a Markov process, a multitude of other Markov processes can be constructed by just modifying the marginal distributions \cite{darsow1992copulas,overbeck2015multivariate}. 

We use the copula representation of the Markov property in  \cite{darsow1992copulas,overbeck2015multivariate}. 
The $n$-dimensional process $\mathbf{X}$ is a Markov process, if and only if, for all $t_1< t_2<\ldots < t_p$, the copula $C_{t_1,\ldots,t_p}$ of $(\mathbf{X}_{t_1},\ldots,\mathbf{X}_{t_p})$ satisfies \cite{overbeck2015multivariate}
\begin{IEEEeqnarray}{rCl}
C_{t_1,\ldots,t_p} = C_{t_1,t_2}\stackrel{C_{t_2}(.)}{\star}C_{t_2,t_3}\stackrel{C_{t_3}(.)}{\star}\ldots\stackrel{C_{t_{p-1}}(.)}{\star}C_{t_{p-1},t_p}. \IEEEeqnarraynumspace
\end{IEEEeqnarray}
Provided that the integral exists for all $\mathbf{x}$, $\mathbf{y}$, $\mathbf{z}$, 
the operator $\stackrel{C(.)}{\star}$ is defined by
\begin{equation}
(A\stackrel{C(\mathbf{z})}{\star}B)(\mathbf{x},\mathbf{y}) = \int_{0}^{\mathbf{z}} A_{,C}(\mathbf{x},\mathbf{r})\cdot B_{C,}(\mathbf{r},\mathbf{y})C(d\mathbf{r}),
\end{equation}
where $A$ is a $(k + n)$-dimensional copula, $B$ is a $(n + l)$-dimensional
copula, $C$ is a $n$-dimensional copula, and $A(\bm{x},d\bm{y}) =A_{,C}(\mathbf{x}, \mathbf{y})C(d\bm{y})$ and $B(d\bm{x},\bm{y})=B_{C,}(\mathbf{x}, \mathbf{y})C(d\bm{x})$ are respectively the derivative of the copula $A(\mathbf{x}, .)$ and $B(., \mathbf{y})$ with respect to the copula $C$.
$A_{,C}$ and $B_{C,}$ are well-defined.
Specifically, for $1$-dimensional Markov process, the copula is expressed by \cite{darsow1992copulas}
\begin{equation}
C_{t_1\ldots{t_n}} = C_{t_1 t_2}\star C_{t_2 t_3}\star\ldots\star C_{t_{n-1}t_n},
\end{equation}
where $C_{t_1\ldots{t_n}}$ is the copula of $\left(X_{t_1},\ldots,X_{t_n}\right)$, $C_{t_{k-1}t_k}$ is the copula of $\left( X_{t_{k-1}}, X_{t_k} \right)$, and $A\star{B}$ is defined as
\begin{multline}
A\star B \left( x_1,\ldots,x_{m+n-1} \right) = \\
\int\limits_{0}^{x_m} 
\frac{\partial A_{,m}(x_1,\ldots,x_{m-1},\xi)}{\partial{\xi}} \frac{\partial B_{1,}(\xi,x_{m+1},\ldots,x_{m+n-1})}{\partial\xi} d\xi,
\end{multline}
for $m$-dimensional copula $A$ and $n$-dimensional copula $B$.

Examples of Markov family copula are Gaussian copula and Fr\'{e}chet copula \cite{overbeck2015multivariate}.
The $2$-dimensional Fr\'{e}chet copula is available in \cite{darsow1992copulas} and the $n$-dimensional extension is available in \cite{overbeck2015multivariate}.

\begin{example}
The $n$-dimensional Gausssian copula is expressed as 
\begin{eqnarray}
C_{\bm{\Sigma}}(\bm{u}) = \Phi_{\bm{\Sigma}}\left( \Phi^{-1}(u_1),\ldots,\Phi^{-1}(u_n) \right),
\end{eqnarray}
where $\Phi_{\bm{\Sigma}}$ denotes the joint distribution of the $n$-dimensional standard normal distribution with linear correlation matrix $\bm{\Sigma}$, and $\Phi^{-1}$ denotes the inverse of the distribution function of the $1$-dimensional standard normal distribution.
\end{example}

The extremely positive dependence, independence, and extremely negative dependence are expressed by copulas.
For $2$-dimensional copula, the extremely positive copula, product copula (independence), and extremely negative copula are defined as $M(x,y) = \min(x,y)$, $P(x,y) = xy$, and $W(x,y) = \max(x+y-1,0)$.

\begin{example}
A convex combination of $M$, $P$, and $W$ is a Markov family copula, i.e.,
\begin{equation}
C_{st} = \alpha(s,t)W + (1-\alpha(s,t)-\beta(s,t))P + \beta(s,t)M,
\end{equation}
if and only if \cite{darsow1992copulas,overbeck2015multivariate}, for $s<u<t$,
\begin{eqnarray}
\alpha(s,t) &=& \beta(s,u)\alpha(u,t) + \alpha(s,u)\beta(u,t), \\
\beta(s,t) &=& \alpha(s,u)\alpha(u,t) + \beta(s,u)\beta(u,t),
\end{eqnarray}
where $\alpha(s,t)\ge{0}$, $\beta(s,t)\ge{0}$, and $\alpha(s,t)+\beta(s,t)\le{1}$. 
For homogeneous case, $\alpha(s,t)=\alpha(t-s)$ and $\beta(s,t)=\beta(t-s)$, a solution is as follows
\begin{eqnarray}
\alpha(h) &=& {e^{-2h}(1 - e^{-h})}/{2}, \\
\beta(h) &=& {e^{-2h}(1 + e^{-h})}/{2}.
\end{eqnarray}
Let $\alpha = e^{-h}$, it's a one-parameter copula \cite{darsow1992copulas}
\begin{equation}
C_\alpha = \frac{\alpha^2(1-\alpha)}{2}W + (1-\alpha^2)P + \frac{\alpha^2(1+\alpha)}{2}M,
\end{equation}
where $-1\le\alpha\le{1}$, if $|\alpha|$ is small, independence is indicated, if $\alpha$ is near $1$, strongly positive dependence is indicated, and if $\alpha$ is near $-1$, strongly negative dependence is indicated.
\end{example}

\subsubsection{Copula Manipulation}

We assume no Granger causality among random parameter processes.
No-Granger causality is a concept initially introduced in econometrics and refers to a multivariate dynamic system in which each variable is determined by its own lagged values and no further information is provided by the lagged values of other variables \cite{cherubini2011copula}. 

\begin{proposition}
For a $n$-dimensional process $\bm{X}$, $\bm{X}^1,\ldots,\bm{X}^{i-1}$, $\bm{X}^{i+1},\ldots,\bm{X}^n$ do not Granger cause $\bm{X}^i$, if \cite{cherubini2010copula,cherubini2011copula}
\begin{equation}
\mathbb{P}\left( X^{i}_{t_{k+1}} \le x | \mathscr{F}_{t_k}^{\bm{X}^1,\ldots,\bm{X}^n} \right) = \mathbb{P}\left( X^{i}_{t_{k+1}} \le x | \mathscr{F}_{t_k}^{\bm{X}^i} \right).
\end{equation}
\end{proposition}

No-Granger causality and Markov property of each process with respect to its natural filtration together imply the Markov structure of the system as a whole \cite{cherubini2010copula,cherubini2011copula}.
However, additional restriction is required for the converse to hold, the $2$-dimensional result is available in \cite{cherubini2011copula}, and the following theorem is an extension to $n$-dimensional case.
We prove the theorem in Appendix \ref{proof-of-theorem-no-granger-causality-biset}.

\begin{theorem}\label{theorem_no_granger_causality_biset}
For a $n$-dimensional Markov process $\bm{X}$ consisting two dimension sets $\overline{\bm{X}}$ and $\underline{\bm{X}}$, $\bm{X}= \overline{\bm{X}} \cup \underline{\bm{X}}$, $\overline{\bm{X}}$ does not Granger cause $\underline{\bm{X}}$, if and only if
\begin{multline}
C_{j,j+1} \left(\bm{u}_{\underline{\bm{X}}_j}, \bm{u}_{\overline{\bm{X}}_j}, \bm{u}_{\underline{\bm{X}}_{j+1}}, \bm{1}_{\bm{u}_{\overline{\bm{X}}_{j+1}}} \right)  \\
= C_{\overline{\bm{X}}_j\underline{\bm{X}}_j} \stackrel{C_{\underline{\bm{X}}_j} \qty(\bm{u}_{\underline{\bm{X}}_j})}{\star} C_{\underline{\bm{X}}_j\underline{\bm{X}}_{j+1}} \left(\bm{u}_{\overline{\bm{X}}_j}, \bm{u}_{\underline{\bm{X}}_{j+1}} \right),
\end{multline}
$\underline{\bm{X}}$ does not Granger cause $\overline{\bm{X}}$, if and only if
\begin{multline}
C_{j,j+1} \left(\bm{u}_{\underline{\bm{X}}_j}, \bm{u}_{\overline{\bm{X}}_j}, \bm{1}_{\bm{u}_{\underline{\bm{X}}_{j+1}}}, {\bm{u}_{\overline{\bm{X}}_{j+1}}} \right)  \\
= C_{\underline{\bm{X}}_j\overline{\bm{X}}_j} \stackrel{C_{\overline{\bm{X}}_j} \qty(\bm{u}_{\overline{\bm{X}}_j})}{\star} C_{\overline{\bm{X}}_j\overline{\bm{X}}_{j+1}} \left(\bm{u}_{\underline{\bm{X}}_j}, \bm{u}_{\overline{\bm{X}}_{j+1}} \right).
\end{multline}
\end{theorem}

\begin{remark}
Specifically, for the wireless channel capacity that is modeled by a multivariate Markov process, let $\overline{\bm{X}}$ and $\underline{\bm{X}}$ represent respectively the uncontrollable and controllable parameters.
The no-Granger causality guarantees that if the uncontrollable and controllable parameters form a multivariate Markov process, the processes of the uncontrollable and controllable parameters are also Markov processes, which is necessary in dependence control because we need to model the uncontrollable parameters with a certain process and to configure the controllable parameters in a certain way based on a certain process.
\end{remark}

A stronger restriction is that all the $1$-dimensional Markov processes do not Granger cause each other, and the results are as follows.
We present the proof in Appendix \ref{proof-of-theorem-no-granger-causality-stronger}.

\begin{theorem}\label{theorem_no_granger_causality}
For a $n$-dimensional Markov process $\bm{X}$ with temporal copula $C_{j,j+1}$ and spatial copula $C_j$,
\begin{IEEEeqnarray}{rCl}
\mathbb{P}\left( X^{i}_{t_{k+1}} \le x | {\bm{X}^1_{t_k},\ldots,\bm{X}^n_{t_k}} \right) = \mathbb{P}\left( X^{i}_{t_{k+1}} \le x | {\bm{X}^i_{t_k}} \right), \IEEEeqnarraynumspace
\end{IEEEeqnarray}
if and only if
\begin{multline}
C_{j,j+1} \left( x_j^{1},\ldots,x_j^n,1,\ldots,x_{j+1}^i,\ldots,1 \right)  \\
= C_j^{,i} \star C_{j,j+1}^i \left( x_j^1,\ldots,x_j^{i-1},x_j^{i+1},\ldots,x_j^n,x_j^i,x_{j+1}^i \right),
\end{multline}
where $C_j^{,i}$ is the reordered spatial copula, and $C_{j,j+1}^i$ is the temporal copula of the $1$-dimensional Markov process $\bm{X}^i$.
\end{theorem}

\begin{example}
For a $2$-dimensional Markov process $\bm{X}$, $\bm{X}^2$ does not Granger cause $\bm{X}^1$, if and only if \cite{cherubini2011copula}
\begin{equation}
C_{j,j+1}(u_1,v_1,u_2,1) = C_{X^2_j,X^1_j} \star C_{X^1_j,X^1_{j+1}}(v_1,u_1,u_2),
\end{equation}
and $\bm{X}^1$ does not Granger cause $\bm{X}^2$, if and only if \cite{cherubini2011copula}
\begin{equation}
C_{j,j+1}(u_1,v_1,1,v_2) = C_{X^1_j,X^2_j} \star C_{X^2_j,X^2_{j+1}}(u_1,v_1,u_2).
\end{equation}
In the special case, if the spatial dependence is expressed by the product copula, then 
\begin{eqnarray}
C_{j,j+1}(u_1,v_1,u_2,1) &=& v_1 C_{X^1_{j}X^1_{j+1}} \left( u_1,u_2 \right),\\
C_{j,j+1}(u_1,v_1,1,v_2) &=& u_1 C_{X^2_{j}X^2_{j+1}} \left( v_1,v_2 \right).
\end{eqnarray}
\end{example}

Since the copula requires continuity by definition, interpolation is needed to construct a copula from the transition matrix of a Markov process \cite{darsow1992copulas}, while it's not needed to calculate the transition matrix from a copula.
The approach to calculate the transition probability of a Markov chain given the copula of the two consecutive levels is summarized in the following theorem.

\begin{theorem}
For a $1$-dimensional Markov process with finite state space $E$ and initial distribution $\bm\varpi$, given the copula between successive levels $C_{j,j+1}$, 
\begin{equation}
\sum_{s_j\le{\bm x}}\bm{\varpi}_{j}({s_j})\bm{P}_j(s_j,s_{j+1}\le\bm{y}) = C_{j,j+1}\left( F_j(\bm{x}), F_{j+1}(\bm{y}) \right),
\end{equation}
where $\bm{x}$ and $\bm{y}$ are the ordered state space vector, the state distribution at $j$ is $\bm{\varpi}_j=\bm{\varpi}\prod_{0\le{k}\le{j}}\bm{P}_k$, and $F_j(s_j) = \sum\bm{\varpi}_j(s_k\le{s_j})$ and $F_{j+1}=\bm{\varpi}_j{\bm{P}_j}$.
Together with the unity property of transition matrix $\sum_{j\in{E}}p_{ij}=1$, $\forall{i}\in{E}$, the transition probabilities $\bm{P}_j$ are obtained.
\end{theorem}

\begin{proof}
For random variables $X$ and $Y$ with the copula $C$ \cite{darsow1992copulas}
\begin{IEEEeqnarray}{rCl}
E\left( I_{Y<y}|X \right)(\omega) &=& C_{1,}\left( F_X(X(\omega)), F_Y(y) \right)\ a.s.,
\end{IEEEeqnarray}
by integrating,
\begin{equation}
\int_{-\infty}^{\bm{x}} \mathbb{P}\left( X_t\le \bm{y} | X_s = \xi \right) d\xi = C\left( F_s(\bm{x}), F_t(\bm{y}) \right).
\end{equation}
The result directly follows.
\end{proof}

\begin{example}
For a $2$-state homogeneous Markov process, the equations are expressed as
\begin{IEEEeqnarray}{rrCl}
  \IEEEyesnumber\IEEEyessubnumber*
  & C\left(F(0),F(0)\right) & = & \pi_0 p_{00},
  \\*
  & C\left(F(1),F(0)\right) & = & \pi_0 p_{00} + \pi_1 p_{10},
  \\*[-0.625\normalbaselineskip]
  \smash{\left\{
      \IEEEstrut[11\jot]
    \right. } \nonumber
\\*[-0.625\normalbaselineskip]
  & C\left(F(1),F(1)\right) & = & \pi_0 \left(p_{00} + p_{01} \right) + \pi_1 \left(p_{10} + p_{11} \right), \IEEEeqnarraynumspace
  \\*
  & C\left(F(0),F(1)\right) & = & \pi_0 \left(p_{00} + p_{01} \right).
\end{IEEEeqnarray}
Given a stationary distribution $[\pi_{0}\ \pi_{1}]$, $F(0)=\pi_0$ and $F(1)=\pi_0+\pi_1$, we obtain the values of $p_{00}$ and $p_{10}$ from the equations, and we further obtain $p_{01}=1-p_{00}$ and $p_{11}=1-p_{10}$ from the unity property.
\end{example}

\begin{figure}
\centering
\includegraphics[width=3.3in]{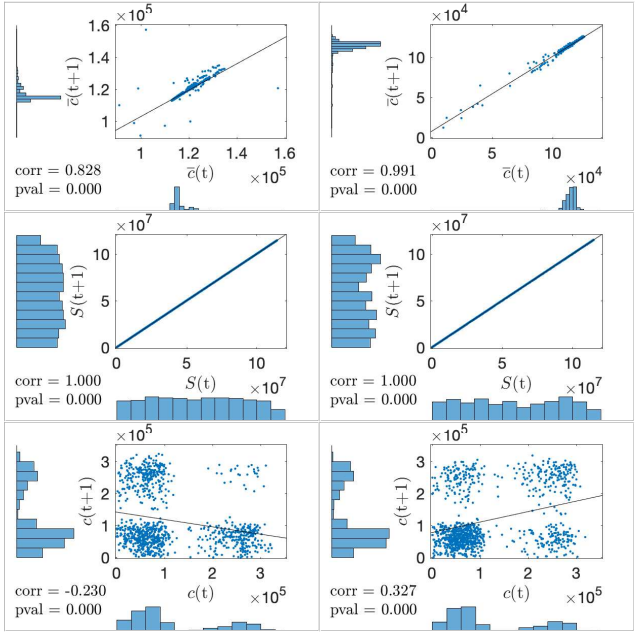}
\caption{Wireless channel capacity of Markov additive Rayleigh channel. 
The uncontrollable parameter is fading with one state and the controllable parameter is power with two states. 
The Markov process is time homogeneous without Granger causality.
The dependence structure is given by Gaussian copula with correlation matrix 
$\bm{\Sigma}=[1\ 0\ -0.5\ 0; 0\ 1\ 0\ 0; -0.5\ 0\ 1\ 0; 0\ 0\ 0\ 1]$ 
as negative dependence (left column), 
$\bm{\Sigma}=[1\ 0\ 0.5\ 0; 0\ 1\ 0\ 0; 0.5\ 0\ 1\ 0; 0\ 0\ 0\ 1]$ 
as positive dependence (right column), initial distribution $\bm{\varpi}= [0.3\ 0.7]$, stationary distribution $\bm{\pi}= [0.3\ 0.7]$. $W=20$kHz and $P/N=[10^{4}~10^{4}; 10~10]$.
$1000$ time slots.
The correlation coefficient and probability value between the time series and lag-1 series are provided.
}
\label{scatter_hist_capacity}
\end{figure}

\begin{algorithm}[!t]
 \caption{Algorithm for Dependence Control}
 \begin{algorithmic}[1]
 \label{dependence_control_algorithm}
 \renewcommand{\algorithmicrequire}{\textbf{Model:}}
 \renewcommand{\algorithmicensure}{\textbf{Result:}}
 \REQUIRE A $n$-dimensional Markov process consisting of $n$ $1$-dimensional Markov processes without Granger causality
 \ENSURE  Transition matrix of the controllable parameter
 \\ 
  \STATE Initialisation: $C_{j,j+1}$, $C_j$, and $\bm{\varpi}$
  \FOR {$j = 0$ to $t-1$}
  \FOR {$1\le i\le{n}$ of interest} 
  \STATE Calculate $\bm{P}_j^i$ and $\bm{\varpi}_{j+1}^i=\bm{\varpi}^i_j{\bm{P}^i_j}$, with\\ $\sum \bm{\varpi}_j^i{\bm{P}_j^i} = C_{j,j+1}^i$
  \ENDFOR
  \ENDFOR
 \RETURN $\bm{P}$ 
 \end{algorithmic} 
\end{algorithm}

The algorithm of dependence control is shown in Algorithm \ref{dependence_control_algorithm}.
It is worth noting that the Markov property is a pure property of copula,
different copula functions provide a way to character the negative or positive dependence, based on which we can calculate the transition matrix of the controllable parameters in the wireless system, e.g., power, and bring their impacts into capacity.

A simulation example is illustrated in Fig. \ref{scatter_hist_capacity}. 
The fading process is independent and the power changes with negative or positive dependence, the result shows that the times series of the instantaneous capacity $C(t)$ exhibits weakly negative dependence or weakly positive dependence, and the impact is manifested in the transient capacity $\overline{C}(t)$.
Since the instantaneous capacity is non-negative, the cumulative capacity $S(t)$ exhibits extremely positive dependence, no matter the negative or positive dependence in the instantaneous capacity.

\section{Concluding Remarks}
\label{conclusion}

This paper initiates the research on dependence control, which transforms the dependence structure of a stochastic process in the system through dependence manipulation to improve the system performance. Specifically, we develop a dependence control theory for wireless channels and define three principles in dependence control, namely the asymptotic measure, the dual potency, and the dependence transformation. To this end, a set of results making use of various mathematical techniques like change of measure, stochastic order, martingale, and copula, are provided. 

While the focus of this paper is on dependence control in wireless channel capacity, many of the obtained results hold for general queueing systems. In the development of the theory, several assumptions are made, which  allow to characterize weak forms of dependence and light-tailed process, and an example is the Markov additive process. We remark that, among the three principles, these assumptions are necessary only for the first. For the second principle, it relies on the assumptions when the dual potency implies the ordering of the asymptotic measure, and for the third principle, the results do not require these assumptions. 
The investigation and extension of the three principles to other forms of dependence and heavy-tailed processes are our future work.

We highlight that the goal of this paper is to pave the way for the development of dependence control particularly as an approach to utilize the hidden resource in wireless channel. We believe the dependence control is a new direction for research, 
and the potential for further development includes additional perspectives on the measure identity, diverse manipulation techniques to transform the dependence structure, and more application scenarios.

\bibliographystyle{ACM-Reference-Format}
\bibliography{main}


\appendix

\section{Proof of Theorem \ref{theorem-decay-rate-delay-backlog}}
\label{proof-of-theorem-decay-rate-delay-backlog}


We only provide proof for the delay result, since the backlog result is a trivial reduction of the delay proof.
The proof is inspired by \cite{glynn1994logarithmic,asmussen2010ruin}, by defining a new change of measure, and by noting the following result, for large enough $n$, 
\begin{equation}
\widetilde{\mathbb{P}}_n \left( \left| \frac{ \mathfrak{S}(n-k)-A(d) + d }{n} -\widetilde{\mu} \right| > \eta \right) \le z^n.
\end{equation}

We first show that $\liminf_{d\rightarrow\infty}\frac{1}{d}\log \mathbb{P}(D>d) \ge -\kappa^{A}(\gamma)$.
Given $\eta>0$ and let $m\equiv m(\eta) = \lfloor d(1+\eta)/\widetilde{\mu} \rfloor +1$.
Then
\begin{IEEEeqnarray}{rCl}
\IEEEeqnarraymulticol{3}{l}{
\mathbb{P}(D>d) \ge \mathbb{P}(\mathfrak{S}(m)>A(d)) 
}\\
&=& \widetilde{\mathbb{E}}_{m} \left[ e^{-\gamma\mathfrak{S}(m) + \kappa_m(\gamma)}; \mathfrak{S}(m) - A(d) + d > d \right] \\
&\ge&  \widetilde{\mathbb{E}}_{m} \left[ e^{-\gamma\mathfrak{S}(m) + \kappa_m(\gamma)}; \frac{\mathfrak{S}(m) - A(d) + d}{m} - \widetilde{\mu} > -\frac{\widetilde{\mu}\eta}{1+\eta} \right] \IEEEeqnarraynumspace\\
&\ge&  \widetilde{\mathbb{E}}_{m} \left[ e^{-\gamma\mathfrak{S}(m) + \kappa_m(\gamma)}; \left| \frac{\mathfrak{S}(m) - A(d) + d}{m} - \widetilde{\mu} \right| < \frac{\widetilde{\mu}\eta}{1+\eta} \right] \\
&\ge&  \widetilde{\mathbb{E}}_{m} \left[ e^{-\gamma \left( \widetilde{\mu}\frac{1+2\eta}{1+\eta}m + A(d) -d \right) + \kappa_m(\gamma)} \right] \nonumber\\
&& \cdot \widetilde{\mathbb{P}}_{m} \left( \left| \frac{\mathfrak{S}(m) - A(d) + d}{m} - \widetilde{\mu} \right| < \frac{\widetilde{\mu}\eta}{1+\eta}  \right) \\
&=& e^{ -\kappa_d^A{\gamma} -\gamma\widetilde{\mu}\frac{1+2\eta}{1+\eta}m + \gamma{d} + \kappa_m(\gamma) } \nonumber\\
&& \cdot \widetilde{\mathbb{P}}_{m} \left( \left| \frac{\mathfrak{S}(m) - A(d) + d}{m} - \widetilde{\mu} \right| < \frac{\widetilde{\mu}\eta}{1+\eta}  \right),
\end{IEEEeqnarray}
where $\widetilde{\mathbb{P}}_m(\cdot)$ goes to $1$ according to Corollary \ref{corollary-for-decay}. Since $\kappa_m(\gamma)/d \rightarrow 0$ and $m/d \rightarrow (1+\eta)/\widetilde{\mu}$, we get
\begin{equation}
\liminf_{d\rightarrow\infty}\frac{1}{d}\log \mathbb{P}(D>d) \ge -\kappa^A(\gamma) -2\eta.
\end{equation}
Letting $\eta\downarrow 0$ yields $\liminf_{d\rightarrow\infty}\frac{1}{d}\log \mathbb{P}(D>d) \ge -\kappa^{A}(\gamma)$.

We then show that $\limsup_{d\rightarrow\infty}\frac{1}{d}\log \mathbb{P}(D>d) \le -\kappa^{A}(\gamma)$.
Let $\tau(d) = \inf\{ n: \mathfrak{S}(n) > A(d) \}$ and $\mathbb{P}(D>d) = \mathbb{P}(\tau(d)<\infty)$, then
\begin{eqnarray}
\mathbb{P}(D>d) = \sum_{n=d}^{\infty} \mathbb{P}(\tau(d) = n) 
= I_1 + I_2 + I_3 + I_4,
\end{eqnarray}
where 
\begin{eqnarray}
I_1 &=& \sum_{n=d}^{n(\delta)} \mathbb{P} (\tau(d)=n), \\
I_2 &=& \sum_{n=n(\delta)+1}^{\lfloor d(1 - \delta)/\widetilde{\mu} \rfloor} \mathbb{P} (\tau(d)=n), \\
I_3 &=& \sum_{n = \lfloor d(1 - \delta)/\widetilde{\mu} \rfloor +1 }^{\lfloor d(1+\delta)/\widetilde{\mu} \rfloor} \mathbb{P} (\tau(d)=n), \\
I_4 &=& \sum_{n = \lfloor d(1+\delta)/\widetilde{\mu} \rfloor +1 }^{\infty} \mathbb{P} (\tau(d)=n), 
\end{eqnarray}
and $n(\delta)$ is chosen such that $\kappa_n(\gamma)/n< \min\{ \delta, (-\log{z})/2 \}$ and 
\begin{equation}
\widetilde{\mathbb{P}}_n \left( \left| \frac{ \mathfrak{S}(n-k)-A(d) + d }{n} -\widetilde{\mu} \right| > \frac{\delta\widetilde{\mu}}{1+\delta} \right) \le z^n,\ \text{for } k \le 1,
\end{equation}
for some $z<1$ and all $n>n(\delta)$. This is possible by Assumption (\ref{aspt-3}) and (\ref{aspt-4}) and Corollary \ref{corollary-for-decay}. 

Note
\begin{eqnarray}
\mathbb{P}(\tau(d)=n) &\le& \mathbb{P}(\mathfrak{S}(n)>A(d)) \\
&=& \widetilde{\mathbb{E}}_{n}\left[ e^{-\gamma \mathfrak{S}(n) + \kappa_n(\gamma) }; \mathfrak{S}(n)>A(d) \right] \\
&\le& e^{-\kappa_d^A(\gamma)} \cdot e^{\kappa_n(\gamma)} \cdot \widetilde{\mathbb{P}}_{n}(\mathfrak{S}(n)>A(d)),
\end{eqnarray}
so that 
\begin{IEEEeqnarray}{rCl}
I_1 &\le& e^{-\kappa_d^A(\gamma)} \sum_{n=d}^{n(\delta)} e^{\kappa_n(\gamma)}, \\
I_2 &\le& e^{-\kappa_d^A(\gamma)} \sum_{n=n(\delta)+1}^{\lfloor d(1 - \delta)/\widetilde{\mu} \rfloor} e^{\kappa_n(\gamma)} \widetilde{\mathbb{P}}_{n}(\mathfrak{S}(n)>A(d)) \\
&\le& e^{-\kappa_d^A(\gamma)} \sum_{n=n(\delta)+1}^{\lfloor d(1 - \delta)/\widetilde{\mu} \rfloor} e^{-n\log{z}/2}  \nonumber\\
&& \cdot \widetilde{\mathbb{P}}_n \left( \left| \frac{ \mathfrak{S}(n)-A(d) + d }{n} -\widetilde{\mu} \right| > \frac{\delta\widetilde{\mu}}{1+\delta} \right) \\
&\le& e^{-\kappa_d^A(\gamma)} \sum_{n=n(\delta)+1}^{\lfloor d(1 - \delta)/\widetilde{\mu} \rfloor} \frac{1}{z^{n/2}}z^n \\
&\le&  e^{-\kappa_d^A(\gamma)} \sum_{n=0}^{\infty} z^{n/2} \\
&=& e^{-\kappa_d^A(\gamma)} \frac{1}{1-z^{1/2}}, \\
I_3 &\le& e^{-\kappa_d^A(\gamma)} \sum_{n = \lfloor d(1 - \delta)/\widetilde{\mu} \rfloor +1 }^{\lfloor d(1+\delta)/\widetilde{\mu} \rfloor} e^{\kappa_n(\gamma)} \\
&\le& e^{-\kappa_d^A(\gamma)} \sum_{n = \lfloor d(1 - \delta)/\widetilde{\mu} \rfloor +1 }^{\lfloor d(1+\delta)/\widetilde{\mu} \rfloor} e^{n\delta} \\
&\le& e^{-\kappa_d^A(\gamma)} \left( \frac{2\delta d}{\widetilde{\mu}} +1 \right) e^{\delta d (1+\delta)/\widetilde{\mu}}.
\end{IEEEeqnarray}
Finally, let $\mathfrak{S}_{n-1}^{n}(d) \equiv \{ \mathfrak{S}(n-1)\le A(d), \mathfrak{S}(n)> A(d) \}$,
\begin{IEEEeqnarray}{rCl}
I_4 &\le&  \sum_{n = \lfloor d(1+\delta)/\widetilde{\mu} \rfloor +1 }^{\infty} \mathbb{P} \left( \mathfrak{S}_{n-1}^{n}(d)  \right) \\
&=& \sum_{n = \lfloor d(1+\delta)/\widetilde{\mu} \rfloor +1 }^{\infty} \widetilde{\mathbb{E}}_{n}\left[ e^{-\gamma \mathfrak{S}(n) + \kappa_n(\gamma) }; \mathfrak{S}_{n-1}^{n}(d) \right] \\
&\le& e^{-\kappa_d^A(\gamma)} \sum_{n = \lfloor d(1+\delta)/\widetilde{\mu} \rfloor +1 }^{\infty}  e^{\kappa_n(\gamma)} \nonumber\\
&& \cdot \widetilde{\mathbb{P}}_n \left( \left| \frac{ \mathfrak{S}(n-1)-A(d) + d }{n} -\widetilde{\mu} \right| > \frac{\delta\widetilde{\mu}}{1+\delta} \right) \\
&\le& e^{-\kappa_d^A(\gamma)}  \sum_{n = \lfloor d(1+\delta)/\widetilde{\mu} \rfloor +1 }^{\infty}  \frac{1}{z^{n/2}}z^n  \\
&\le& e^{-\kappa_d^A(\gamma)} \frac{1}{1-z^{1/2}}.
\end{IEEEeqnarray}
By Assumption (\ref{aspt-1}) and Proposition (\ref{proposition-1}), we get
\begin{equation}
\limsup_{d\rightarrow\infty} \frac{1}{d}{\log \mathbb{P}(D>d)} \le -\kappa^A(\gamma) + \frac{\delta(1+\delta)}{\widetilde{\mu}}.
\end{equation}
Letting $\delta\downarrow 0$ yields $\limsup_{d\rightarrow\infty}\frac{1}{d}\log \mathbb{P}(D>d) \le -\kappa^{A}(\gamma)$.


\section{Proof of Theorem \ref{theorem-equivalent-converge-in-probability}}
\label{proof-of-theorem-equivalent-converge-in-probability}


Let $0<\theta<\epsilon$, where $\epsilon$ is as in Assumption.
Note $\mathbb{E}\left[ e^{\theta(\mathfrak{S}(n) - \mathfrak{S}(n-k))} \right] $ $< \infty$ for all $|\theta|<\delta$ for some $\delta>0$ by Assumption (\ref{aspt-2}).

According to Chernoff bound,
\begin{IEEEeqnarray}{rCl}
\IEEEeqnarraymulticol{3}{l}{
\widetilde{\mathbb{P}}_n \left(  \frac{ \mathfrak{S}(n-k)- \mathfrak{S}(d) }{n}  - \widetilde{\mu} > \eta \right) 
}\\
&\le& e^{-\theta n\left( \widetilde{\mu} + \eta \right) } \widetilde{\mathbb{E}}_{n} \left[ e^{\theta ( \mathfrak{S}(n-k) - \mathfrak{S}(d) )} \right] \\
&=& e^{-\theta n\left( \widetilde{\mu} + \eta \right) } {\mathbb{E}}_{n} \left[ e^{\theta ( \mathfrak{S}(n-k) - \mathfrak{S}(d) )} \cdot e^{\gamma\mathfrak{S}(n)-\kappa_n(\gamma)} \right] \\
&=& e^{-\theta n\left( \widetilde{\mu} + \eta \right)  - \kappa_n(\gamma) }  {\mathbb{E}}_{n} \left[ e^{ (\theta+\gamma)\mathfrak{S}(n) -\theta \mathfrak{S}(d) -\theta \mathfrak{S}(n-k,n) } \right] \IEEEeqnarraynumspace\\
&\le& e^{-\theta n\left( \widetilde{\mu} + \eta \right)  - \kappa_n(\gamma)  }   \left[ \left[ {\mathbb{E}}_{n} e^{ \hat{p}p(\theta+\gamma)\mathfrak{S}(n) } \right]^{1/\hat{p}} \left[ {\mathbb{E}}_{n} e^{ -\hat{q}p\theta\mathfrak{S}(d) } \right]^{1/\hat{q}} \right]^{1/p}  \nonumber\\
&& \cdot \left[ \mathbb{E}_{n} e^{ - q\theta \mathfrak{S}(n-k,n) } \right]^{1/q} \\
&=& e^{ -\theta n\left( \widetilde{\mu} + \eta \right)  - \kappa_n(\gamma) + \kappa_{n}(\hat{p}p(\theta+\gamma))/({\hat{p}p}) } \left[ {\mathbb{E}}_{n} e^{ -\hat{q}p\theta\mathfrak{S}(d) }  \right]^{1/(\hat{q}{p})}  \nonumber\\
&& \cdot  \left[ \mathbb{E}_{n} e^{ - q\theta \mathfrak{S}(n-k,n) } \right]^{1/q},
\end{IEEEeqnarray}
where we used H\"{o}lder's inequality twice, for positive $p$ and $q$ with $p^{-1} + q^{-1}=1$, and $\hat{p}$ and $\hat{q}$ with $\hat{p}^{-1} + \hat{q}^{-1}=1$, and we choose $p$ and $\hat{p}$ close enough to $1$ and $\theta$ close enough to $0$ that $|\hat{p}p(\theta+\gamma)-\gamma|<\epsilon$ and $|-\hat{q}p\theta-\gamma|<\epsilon$. 
Particularly, for $k=0$ or $d=0$, the proof needs to use H\"{o}lder's inequality only once; for $k=0$ and $d=0$, the proof needs no H\"{o}lder's inequality.

By Assumption (\ref{aspt-1}), $ {\mathbb{E}}_{n} e^{ -\hat{q}p\theta\mathfrak{S}(d) } <\infty$, and by Assumption (\ref{aspt-2}), $\mathbb{E}_{n} \left[ e^{ - q\theta\left( \mathfrak{S}(n) - \mathfrak{S}(n-k) \right) } \right]^{1/q} < \infty$ for large $n$, we get
\begin{multline}
\limsup\limits_{n\rightarrow\infty}\frac{1}{n}  \log \widetilde{\mathbb{P}}_{n} \left(  \frac{ \mathfrak{S}(n-k)- \mathfrak{S}(d)  }{n}  - \widetilde{\mu} > \eta \right) \\
\le \kappa(\hat{p}p(\theta+\gamma))/(\hat{p}p) - \kappa(\gamma) - \theta\left( \widetilde{\mu} + \eta \right),
\end{multline}
by Taylor expansion, it is easy to see that the right hand side can be chosen strictly negative by taking $p$ and $\hat{p}$ close enough to $1$ and $\theta$ close enough to $0$.
This establishes $\widetilde{\mathbb{P}}_n (  { \mathfrak{S}(d,n-k) }/{n}  - \widetilde{\mu} > \eta ) \le z^n$, and the corresponding $\widetilde{\mathbb{P}}_n (  { \mathfrak{S}(d,n-k) }/{n}  - \widetilde{\mu} < -\eta ) \le z^n$ follows by symmetry.


\section{Proof of Theorem \ref{theorem-random-multiplexing}}
\label{proof-of-theorem-random-multiplexing}


The first result is proved in \cite{shaked2007stochastic}. For the second result,
since
\begin{multline}
\mathbb{E} \left[ \phi\left( \sum_{j=1}^{N_1} X_{j,1}, \ldots, \sum_{j=1}^{N_m} X_{j,m}  \right) \Big| \left( N_1, \ldots, N_m \right) = \left( n_1, \ldots, n_m \right) \right] \\
 \le \mathbb{E} \left[ \phi \left( \sum_{j=1}^{N_1} Y_{j,1}, \ldots, \sum_{j=1}^{N_m} Y_{j,m} \right) \Big| \left( N_1, \ldots, N_m \right) = \left( n_1, \ldots, n_m \right) \right],
\end{multline}
for any supermodular function $\phi$, thus
\begin{multline}
\mathbb{E} \left[ \phi\left( \sum_{j=1}^{N_1} X_{j,1}, \ldots, \sum_{j=1}^{N_m} X_{j,m}  \right)  \right] 
 \le \mathbb{E} \left[ \phi \left( \sum_{j=1}^{N_1} Y_{j,1}, \ldots, \sum_{j=1}^{N_m} Y_{j,m} \right)\right].
\end{multline}
Considering the first and second results, the third result is obvious.


\section{Proof of Theorem \ref{theorem_performance_bound}}
\label{proof-of-theorem-performance-bound}


The idea of the proof is to find a likelihood ratio martingale of the process $A(d,t)-S(0,t)$ for delay and $A(t)-S(t)$ for backlog, change the measure, by the likelihood ratio identity, we obtain a likelihood ratio representation of the probability in the new measure.

We provide the full proof of the delay tail probability, 
\begin{equation}
\mathbb{P} ( D > d ) = \mathbb{P}\left \{ \sup_{t\ge{d}}\left( A(d,t) - S(0,t) \right) > 0 \right\}.
\end{equation}

Recall the definition of the Markov additive process
$
\mathbb{E}[ f(S({t+s})-S(t))g(J_{t+s})|\mathscr{F}_t ] = \mathbb{E}_{J_t,0}[ f(S(s))g(J_s) ],
$
which indicates that the time shift of the process is only dependent on the state at the shift epoch, specifically, for $\theta>0$, the likelihood ratio martingale of the arrival process $A(d,t)$ is expressed as
\begin{equation}
L^{A}_{t-d} \circ \theta_{d} = \frac{h_{J_t}^{A}{(\theta)}}{h_{J_d}^{A}{(\theta)}} e^{\theta A(d,t) - (t-d)\kappa^{A}(\theta)},
\end{equation}
where $\theta_d$ is the shift operator; and the likelihood ratio martingale of the service process $-S(t)$ is 
\begin{equation}
L^{-S}_{t} = \frac{h_{J_t}^{-S}{(\theta)}}{h_{J_0}^{-S}{(\theta)}} e^{-\theta S(0,t) - t\kappa^{-S}(\theta)}.
\end{equation}
Assume the arrival process and the service process are independent, then the product of the martingales 
\begin{equation}
L^{A-S}_{d,t} = \left( L^{A}_{t-d} \circ \theta_{d} \right) \cdot L^{-S}_{t}
\end{equation}
is also a martingale \cite{cherny2006some}, and
\begin{equation}
\mathbb{E} \left[ L^{A-S}_{d,t} \right] = \mathbb{E}\left[ L^{A}_{t-d} \circ \theta_{d} \right] \cdot \mathbb{E}\left[ L^{-S}_{t} \right] = 1.
\end{equation}

Define the stopping time $\tau(d) = \inf\{t\ge{d}: A(d,t) - S(0,t) > 0 \}$.
Let $H(\theta) = \frac{h_{J_d}^{A}{(\theta)}}{h_{J_{\tau(d)}}^{A}{(\theta)}} \frac{h_{J_0}^{-S}{(\theta)}}{h_{J_{\tau(d)}}^{-S}{(\theta)}}$.
The delay tail probability, conditional on the initial state ${\bm{J}}_{d,0}=\bm{i}$, i.e., $\left\{ J^{A}_d, J^{-S}_0 \right\}=\left\{ i^A, i^{-S} \right\}$, is expressed as
\begin{IEEEeqnarray}{rCl}
\IEEEeqnarraymulticol{3}{l}{
\mathbb{P}_{\bm i}(D >d)  = \mathbb{P}_{\bm i}(\tau(d)<\infty) }\\
&=& \widetilde{\mathbb{E}}_{\bm i} \left[ H(\theta)  e^{ -\theta\xi_{\tau(d)} + (\tau(d)-d)\kappa^{A}(\theta) + \tau(d)\kappa^{-S}(\theta) };\ \tau(d) < \infty \right], \IEEEeqnarraynumspace
\end{IEEEeqnarray}
where $\theta$ is the root to the stability equation
\begin{equation}
\kappa^{A}(\theta) + \kappa^{-S}(\theta) = 0,
\end{equation}
and $\xi_{\tau(d)}>0$ is the overshoot at the hitting time, which is bounded by
\begin{equation}
0< \xi_{\tau(d)} < A(\tau(d)-1,\tau(d)).
\end{equation}

The delay upper bound is expressed as
\begin{IEEEeqnarray}{rCl}
\IEEEeqnarraymulticol{3}{l}{
\mathbb{P}_{\bm i}(D >d)  = \mathbb{P}_{\bm i}(\tau(d)<\infty) }\\
&\le& \widetilde{\mathbb{E}}_{\bm i} \left[ H(\theta)  e^{ (\tau(d)-d)\kappa^{A}(\theta) + \tau(d)\kappa^{-S}(\theta) };\ \tau(d) < \infty \right] \\
&\le& H_{+} \cdot h_{J_0}^{-S}(\theta) \widetilde{\mathbb{E}}_{\bm i} \left[ e^{ (\tau(d)-d)\kappa^{A}(\theta) + \tau(d)\kappa^{-S}(\theta) };\ \tau(d) < \infty \right] \\
&=& H_{+} \cdot h_{J_0}^{-S}(\theta) \cdot e^{-d  \kappa^{A}(\theta)},
\end{IEEEeqnarray}
where
$H_{+} = \frac{\max_{j\in{E}}h_{j}^{A}{(\theta)}}{\min_{j\in{E}}h_{j}^{A}{(\theta)}} \cdot \frac{1}{\min_{j\in{E'}}h_{j}^{-S}(\theta)}$.

The delay lower bound is expressed as
\begin{IEEEeqnarray}{rCl}
\IEEEeqnarraymulticol{3}{l}{
\mathbb{P}_{\bm i}(D >d)  = \mathbb{P}_{\bm i}(\tau(d)<\infty) }\\
&\ge& \widetilde{\mathbb{E}}_{\bm i} \left[ H(\theta)  e^{ -\theta A(\tau(d)-1,\tau(d)) -d\kappa^{A}(\theta) };\ \tau(d) < \infty \right] \\
&\ge& \widetilde{\mathbb{E}}_{\bm i} \left[ e^{ -\theta A(\tau(d)-1,\tau(d)) };\ \tau(d) < \infty \right] \cdot \hat{H}_{-} \cdot  h_{J_0}^{-S}(\theta) \cdot e^{-d  \kappa^{A}(\theta)}, \IEEEeqnarraynumspace
\end{IEEEeqnarray}
where
$\hat{H}_{-} =  \frac{\min_{j\in{E}}h_{j}^{A}{(\theta)}}{\max_{j\in{E}}h_{j}^{A}{(\theta)}} \cdot \frac{1}{\max_{j\in{E'}}h_{j}^{-S}(\theta)}$ 
and
\begin{IEEEeqnarray}{rCl}
\IEEEeqnarraymulticol{3}{l}{
\widetilde{\mathbb{E}}_{\bm i} \left[ e^{ -\theta A(\tau(d)-1,\tau(d)) };\ \tau(d) < \infty \right] }\\
&=& {\mathbb{E}}_{i^A} \left[ e^{ -\theta A(\tau(d)-1,\tau(d)) } \cdot \left( L^{A}_{\tau(d)-d}\circ \theta_{d}  \right);\ \tau(d) < \infty \right] \\
&=& {\mathbb{E}}_{i^A} \left[ \frac{h_{J_{\tau(d)}}^A}{h_{J_{\tau(d)-1}}^A} \cdot \left( L^{A}_{(\tau(d)-1)-d}\circ \theta_{d}  \right) \cdot  e^{-\kappa^{A}(\theta)};\ \tau(d) < \infty \right] \IEEEeqnarraynumspace \\
&\ge& \frac{\min_{j\in{E}} h_j^{A}(\theta) }{\max_{j\in{E}}  h_j^{A}(\theta)} \cdot e^{-\kappa^{A}(\theta)},
\end{IEEEeqnarray}
where the first equality is due to the assumption of independence between the arrival process and service process, and the last inequality follows that $\left( L^{A}_{(\tau(d)-1)-d}\circ \theta_{d}  \right)$ is a mean-one martingale.


The backlog tail probability, conditional on the initial state ${\bm{J}}_{0,0}=\bm{i}$, i.e., $\left\{ J^{A}_0, J^{-S}_0 \right\}=\left\{ i^A, i^{-S} \right\}$, is expressed as
$\mathbb{P}_{\bm i} ( B > b ) = \mathbb{P}_{\bm i} \{ \sup_{t\ge{0}} $ $(A(t)-S(t)) \ge b \}$.
Consider the likelihood ratio process
$
L_{t}^{A-S} = L^{A}_{t} \cdot L_{t}^{-S},
$
which is a mean-one martingale.
Define the stopping time $\tau(b) = \inf\{t\ge{d}: A(t) - S(t) > b \}$ and note the overshoot at the hitting time
$
b < \xi_{\tau{(b)}} < b + A(\tau(b)-1,\tau(b)).
$
Following the same process of change of measure in the proof of delay, 
the proof of the backlog results follows analogically.


\section{Proof of Theorem \ref{theorem-timely-performance}}
\label{proof-of-theorem-timely-performance}


The proof follows two phases, in the first phase, we provide the condition that the inequalities hold, in the second phase, we provide the setting of $y$ that satisfies the condition.

First, we prove that the inequalities hold under a condition on $\kappa^{A}(\theta) + \kappa^{-S}(\theta)$.
Let $H(\theta) = \frac{h_{J_d}^{A}{(\theta)}}{h_{J_{\tau(d)}}^{A}{(\theta)}} \frac{h_{J_0}^{-S}{(\theta)}}{h_{J_{\tau(d)}}^{-S}{(\theta)}}$.
For any $\theta>0$, $\kappa^{A}(\theta) + \kappa^{-S}(\theta)>0$.
\begin{IEEEeqnarray}{rCl}
\IEEEeqnarraymulticol{3}{l}{
P_{\bm i}(D(t) >d; t\le yd) }\\
&=& \widetilde{\mathbb{E}}_{\bm i} \left[ H(\theta)  e^{ -\theta\xi_{\tau(d)} + (\tau(d)-d)\kappa^{A}(\theta) + \tau(d)\kappa^{-S}(\theta) };\ \tau(d)\le yd \right] \IEEEeqnarraynumspace\\
&\le& \widetilde{\mathbb{E}}_{\bm i} \left[ H(\theta)  e^{ (\tau(d)-d)\kappa^{A}(\theta) + \tau(d)\kappa^{-S}(\theta) };\ \tau(d)\le yd \right] \\
&\le& H_{+} h_{J_0}^{-S}(\theta) \widetilde{\mathbb{E}}_{\bm i} \left[ e^{ (\tau(d)-d)\kappa^{A}(\theta) + \tau(d)\kappa^{-S}(\theta) };\ \tau(d)\le yd \right] \\
&\le& H_{+} h_{J_0}^{-S}(\theta) e^{-d (-y\kappa^{-S}(\theta) -(y-1)\kappa^{A}(\theta))}.
\end{IEEEeqnarray}
For any $\theta>0$, $\kappa^{A}(\theta) + \kappa^{-S}(\theta)<0$.
\begin{IEEEeqnarray}{rCl}
\IEEEeqnarraymulticol{3}{l}{
P_{\bm i}(D>d) - P_{\bm i}(D(t) >d; t\le yd) }\\
&=& \widetilde{\mathbb{E}}_{\bm i} \left[ H(\theta)  e^{ -\theta\xi_{\tau(d)} + (\tau(d)-d)\kappa^{A}(\theta) + \tau(d)\kappa^{-S}(\theta) };\ yd < \tau(d) <\infty \right] \IEEEeqnarraynumspace\\
&\le& H_{+} h_{J_0}^{-S}(\theta) \widetilde{\mathbb{E}}_{\bm i} \left[ e^{ (\tau(d)-d)\kappa^{A}(\theta) + \tau(d)\kappa^{-S}(\theta) };\ yd < \tau(d) <\infty \right] \\
&\le& H_{+} h_{J_0}^{-S}(\theta) e^{-d (-y\kappa^{-S}(\theta) -(y-1)\kappa^{A}(\theta))}.
\end{IEEEeqnarray}

Second, we link $y$ to the $\kappa^{A}(\theta) + \kappa^{-S}(\theta)$ condition. Denote
\begin{equation}
\theta_y = -y\kappa^{-S}(\theta) -(y-1)\kappa^{A}(\theta),
\end{equation}
which is a concave function of $\theta$. Thus, for any fixed $y>1$, the optimal $\theta^\ast$ to maximize $\theta_y$ is the root to the derivative equation $\dot{\theta}_y=0$, i.e.,
\begin{equation}
\theta^\ast = \left\{ \theta: y{\dot\kappa}^{-S}(\theta)=-(y-1){\dot\kappa}^{A}(\theta) \right\}.
\end{equation}
Consider the equation
\begin{equation}
\frac{1}{y} = 1 + \frac{{\dot\kappa}^{-S}(\theta)}{{\dot\kappa}^{A}(\theta)},
\end{equation}
since
\begin{equation}
\frac{\partial}{\partial\theta} \left( \frac{{\dot\kappa}^{-S}(\theta)}{{\dot\kappa}^{A}(\theta)} \right)
= \frac{ {\ddot\kappa}^{-S}(\theta) \cdot {\dot\kappa}^{A}(\theta) - {\ddot\kappa}^{A}(\theta) \cdot {\dot\kappa}^{-S}(\theta) }{ \left[ {\dot\kappa}^{A}(\theta) \right]^{2} } \ge 0,
\end{equation}
which indicates that the decrease of $y$ maps to the increase of $\theta$, it follows, if $y<\frac{{\dot{\kappa}}^{A}(\gamma)}{{\dot\kappa}^{A}(\gamma) + {\dot\kappa}^{-S}(\gamma)}$, then $\theta>\gamma$ and $\kappa^{A}(\theta) + \kappa^{-S}(\theta)>0$, vice versa.

The proof of backlog follows analogically. Specifically, for the second phase,
denote
\begin{equation}
\theta_y = \theta - y\left( \kappa^{A}(\theta) + \kappa^{-S}(\theta) \right),
\end{equation}
which is a concave function of $\theta$. Thus, for any fixed $y>0$, the optimal $\theta^\ast$ to maximize $\theta_y$ is the root to the derivative equation $\dot{\theta}_y=0$, i.e.,
\begin{equation}
\theta^\ast = \left\{ \theta: y\left( \dot\kappa^{A}(\theta) + \dot\kappa^{-S}(\theta) \right) =1 \right\}.
\end{equation}
Consider the equation
\begin{equation}
\frac{1}{y} = {{\dot\kappa}^{A}(\theta) + {\dot\kappa}^{-S}(\theta)},
\end{equation}
since
\begin{equation}
\frac{\partial}{\partial\theta} \left({\dot\kappa}^{A}(\theta) + {\dot\kappa}^{-S}(\theta) \right) =  {\ddot\kappa}^{A}(\theta) + {\ddot\kappa}^{-S}(\theta) \ge 0,
\end{equation}
which indicates that the decrease of $y$ maps to the increase of $\theta$, it follows, if $y<\frac{1}{{\dot\kappa}^{A}(\gamma) + {\dot\kappa}^{-S}(\gamma)}$, then $\theta>\gamma$ and $\kappa^{A}(\theta) + \kappa^{-S}(\theta)>0$, vice versa.


\section{Proof of Theorem \ref{theorem_no_granger_causality_biset}}
\label{proof-of-theorem-no-granger-causality-biset}


Since
\begin{IEEEeqnarray}{rCl}
\IEEEeqnarraymulticol{3}{l}{
\mathbb{P} \left( \underline{\bm{X}}_{j+1} \le \bm{x} | {\bm{X}}_j \right)  
}\\
=  {C_{j,j+1}}_{C_{\underline{\bm{X}}_j\overline{\bm{X}}_j}, } \left( F_{\underline{\bm{X}}_j}(\underline{\bm{X}}_j), F_{\overline{\bm{X}}_j}(\overline{\bm{X}}_j), F_{\underline{\bm{X}}_{j+1}}(\bm{x}), \bm{1}_{F_{\overline{\bm{X}}_{j+1}}} \right), \IEEEeqnarraynumspace
\end{IEEEeqnarray}
\begin{IEEEeqnarray}{rCl}
\IEEEeqnarraymulticol{3}{l}{
\mathbb{P} \left( \underline{\bm{X}}_{j+1} \le \bm{x} | {\underline{\bm{X}}}_j \right) 
} \\
=  {C_{j,j+1}}_{ C_{\underline{\bm{X}}_j}, } \left(F_{\underline{\bm{X}}_j}(\underline{\bm{X}}_j), \bm{1}_{F_{\overline{\bm{X}}_j}}, F_{\underline{\bm{X}}_{j+1}}(\bm{x}), \bm{1}_{F_{\overline{\bm{X}}_{j+1}}} \right), \IEEEeqnarraynumspace
\end{IEEEeqnarray}
the no-Granger causality holds, if and only if
\begin{multline}
{C_{j,j+1}}_{C_{\underline{\bm{X}}_j\overline{\bm{X}}_j}, } \left(\bm{u}_{\underline{\bm{X}}_j}, \bm{u}_{\overline{\bm{X}}_j}, \bm{u}_{\underline{\bm{X}}_{j+1}}, \bm{1}_{\bm{u}_{\overline{\bm{X}}_{j+1}}} \right) \\
=
{C_{j,j+1}}_{ C_{\underline{\bm{X}}_j}, } \left(\bm{u}_{\underline{\bm{X}}_j}, \bm{1}_{\bm{u}_{\overline{\bm{X}}_j}}, \bm{u}_{\underline{\bm{X}}_{j+1}}, \bm{1}_{\bm{u}_{\overline{\bm{X}}_{j+1}}} \right). 
\end{multline}
By integrating, we obtain
\begin{IEEEeqnarray}{rCl}
\IEEEeqnarraymulticol{3}{l}{
C_{j,j+1} \left(\bm{u}_{\underline{\bm{X}}_j}, \bm{u}_{\overline{\bm{X}}_j}, \bm{u}_{\underline{\bm{X}}_{j+1}}, \bm{1}_{\bm{u}_{\overline{\bm{X}}_{j+1}}} \right)
}\IEEEeqnarraynumspace\\
&=& \int_{\bm{0}}^{\bm{u}_{\underline{\bm{X}}_j}} {C_{\overline{\bm{X}}_j\underline{\bm{X}}_j}}_{, C_{\underline{\bm{X}}_j}} \left( {\bm{u}_{\overline{\bm{X}}_j}},  {\bm{u}_{\underline{\bm{X}}}} \right) \nonumber\\
&&
\cdot  {C_{\underline{\bm{X}}_j\underline{\bm{X}}_{j+1}}}_{ C_{\underline{\bm{X}}_j}, } \left(  {\bm{u}_{\underline{\bm{X}}}}, {\bm{u}_{\underline{\bm{X}}_{j+1}}} \right) C_{\underline{\bm{X}}_j}\qty(d{{\bm{u}_{\underline{\bm{X}}}}})  \IEEEeqnarraynumspace\\
&=& C_{\overline{\bm{X}}_j\underline{\bm{X}}_j} \stackrel{C_{\underline{\bm{X}}_j} (\bm{u}_{\underline{\bm{X}}_j})}{\star} C_{\underline{\bm{X}}_j\underline{\bm{X}}_{j+1}} \left(\bm{u}_{\overline{\bm{X}}_j}, \bm{u}_{\underline{\bm{X}}_{j+1}} \right). 
\end{IEEEeqnarray}
The other result follows analogically.


\section{Proof of Theorem \ref{theorem_no_granger_causality}}
\label{proof-of-theorem-no-granger-causality-stronger}


The proof follows analogically from Theorem \ref{theorem_no_granger_causality_biset}.
Since
\begin{IEEEeqnarray}{rCl}
\IEEEeqnarraymulticol{3}{l}{
\mathbb{P} \left( X^{i}_{t_{j+1}} \le x_{j+1}^i | {\bm{X}^1,\ldots,\bm{X}^n} \right)  
}\\ 
= \frac{\frac{\partial^n}{\partial{\bm{x}_j}}C_{j,j+1}(x_j^{1},\ldots,x_j^n,1,\ldots,x_{j+1}^i,\ldots,1) } {\frac{\partial^n}{\partial{\bm{x}_j}}C_{j,j+1}(x_j^1,\ldots,x_j^n,1,\ldots,1) }, \IEEEeqnarraynumspace
\end{IEEEeqnarray}
and 
\begin{IEEEeqnarray}{rCl}
\IEEEeqnarraymulticol{3}{l}{
\mathbb{P} \left( X^{i}_{t_{j+1}} \le x_{j+1}^i | {\bm{X}_j^i} \right) 
} \\ 
= \frac{\partial}{\partial{{x}_j^i}}C_{j,j+1}^{i}(1,\ldots,x_j^i,\ldots,1,1,\ldots,x_{j+1}^i,\ldots,1), \IEEEeqnarraynumspace
\end{IEEEeqnarray}
the no-Granger causality holds if and only if
\begin{IEEEeqnarray}{rCl}
\IEEEeqnarraymulticol{3}{l}{
{\frac{\partial^n}{\partial{\bm{x}_j}}C_{j,j+1}(x_j^{1},\ldots,x_j^n,1,\ldots,x_{j+1}^i,\ldots,1) }
} \\
&=& {\frac{\partial^n}{\partial{\bm{x}_j}}C_{j,j+1}(x_j^1,\ldots,x_j^n,1,\ldots,1) } \nonumber\\ 
&& \times { \frac{\partial}{\partial{{x}_j^i}}C_{j,j+1}^{i}(1,\ldots,x_j^i,\ldots,1,1,\ldots,x_{j+1}^i,\ldots,1) }. \IEEEeqnarraynumspace
\end{IEEEeqnarray}
By integrating, we obtain
\begin{IEEEeqnarray}{rCl}
\IEEEeqnarraymulticol{3}{l}{
C_{j,j+1}(x_j^{1},\ldots,x_j^n,1,\ldots,x_{j+1}^i,\ldots,1)
} \\
&=& \int_{0}^{x_j^i}\frac{\partial}{\partial{x^i}}C_j(x^1_j,\ldots,x^i,\ldots,x_j^n) \frac{\partial}{\partial{x^i}}C_{j,j+1}^{i}(x^i,x_{j+1}^i) d x^i \IEEEeqnarraynumspace
\\
&=& \int_{0}^{x_j^i}\frac{\partial}{\partial{x^i}}C_j^{{,i}}(x_j^1,\ldots,x_j^{i-1},x_j^{i+1},\ldots,x_j^n,x_j^i) \nonumber\\ 
&& \times \frac{\partial}{\partial{x^i}}C_{j,j+1}^{i}(x^i,x_{j+1}^i) d x^i \\
&=& C_j^{,i} \star C_{j,j+1}^i(x_j^1,\ldots,x_j^{i-1},x_j^{i+1},\ldots,x_j^n,x_j^i,x_{j+1}^i).
\end{IEEEeqnarray}
This completes the proof.


\end{document}